\documentclass[dvipsnames,letterpaper,11pt]{article}
\usepackage[margin=1in]{geometry}
\usepackage[utf8]{inputenc}
\usepackage{graphicx,fullpage,paralist}
\usepackage{amsmath,amssymb,tabu,amsthm}
\usepackage{dsfont}
\usepackage{bm}
\usepackage{array}
\newcolumntype{C}{>{$}c<{$}}
\usepackage{comment,hyperref}
\usepackage{booktabs}
\usepackage{subcaption}
\usepackage[font={footnotesize}]{caption}
\usepackage{mathtools}
\usepackage{thmtools}
\usepackage{tikz}
\usepackage{tikz-network}
\usepackage{pgfplots}
\usepackage{varwidth}
\usepackage{xstring}
\usepackage{xcolor}
\pgfplotsset{compat=1.10}
\usepgfplotslibrary{fillbetween}
\usetikzlibrary{backgrounds}
\usetikzlibrary{patterns}
\usetikzlibrary{arrows}
\usetikzlibrary{calc}
\usetikzlibrary{matrix,decorations.pathreplacing,positioning,fit}
\usepackage{enumitem}
\usepackage{blkarray}
\usepackage{csquotes}
\usepackage{stmaryrd}
\usepackage{natbib}
\usepackage[skip=2pt, indent=20pt]{parskip}

\definecolor{color0}{RGB}{230,159,0}
\definecolor{color1}{RGB}{86,180,233}
\definecolor{color2}{RGB}{0,158,115}
\definecolor{color3}{RGB}{240,228,66}
\definecolor{color4}{RGB}{0,114,178}
\definecolor{color5}{RGB}{213,94,0}
\definecolor{color6}{RGB}{204,121,167}

\newcommand{\calC}{\mathcal{C}}

\newcommand{\RR}{\mathbb{R}}
\newcommand{\NN}{\mathbb{N}}
\newcommand{\ZZ}{\mathbb{Z}}

\newcommand{\EE}{\mathbb{E}}
\renewcommand{\Pr}{\mathbb{P}}

\newcommand{\OPT}{\mathrm{OPT}}

\newcommand{\R}{\mathbf{R}}

\newcommand{\barnk}{\bar{n}_k}

\newcommand{\kipi}{{\normalfont\text{$(k,n)$-SPI}}}
\newcommand{\ssap}{\normalfont\text{SSAP}}
\newcommand{\nls}{\normalfont\text{NLS}}
\newcommand{\dummylabel}[2]{\def\@currentlabel{#2}\label{#1}}

\newcommand{\mycomment}[1]{}

\newtheorem{theorem}{Theorem}
\newtheorem{lemma}{Lemma}

\newtheorem{proposition}{Proposition}
\newtheorem{claim}{Claim}

\usepackage[textsize=tiny,textwidth=2cm]{todonotes}

\newcommand{\vvcom}[1]{\todo[color=blue!25!white,inline]{Victor: #1}}

\newcommand{\spcom}[1]{\todo[color=red!25!white,inline]{Sebastian: #1}}
\tikzstyle{component}=[draw opacity=0.4,draw=black,line width=1.0cm,line cap=round,line join=round]
\newcommand{\com}[1]{\textcolor{blue}{#1}}

\usepackage{multirow}
\usepackage[noend]{algpseudocode}
\usepackage{algorithm,algorithmicx}

\newlength{\algofontsize}
\hypersetup{
	colorlinks,
	linkcolor={red!50!black},
	citecolor={blue!50!black},
	urlcolor={blue!80!black}
}

\def\equationautorefname~#1\null{(#1)\null}

\usepackage{etoolbox}
\makeatletter
\patchcmd{\hyper@makecurrent}{%
    \ifx\Hy@param\Hy@chapterstring
        \let\Hy@param\Hy@chapapp
    \fi
}{%
    \iftoggle{inappendix}{
        \@checkappendixparam{chapter}%
        \@checkappendixparam{section}%
        \@checkappendixparam{subsection}%
        \@checkappendixparam{subsubsection}%
        \@checkappendixparam{paragraph}%
        \@checkappendixparam{subparagraph}%
    }{}%
}{}{\errmessage{failed to patch}}

\newcommand*{\@checkappendixparam}[1]{%
    \def\@checkappendixparamtmp{#1}%
    \ifx\Hy@param\@checkappendixparamtmp
        \let\Hy@param\Hy@appendixstring
    \fi
}
\makeatletter

\newtoggle{inappendix}
\togglefalse{inappendix}

\apptocmd{\appendix}{\toggletrue{inappendix}}{}{\errmessage{failed to patch}}

\begin{document}
\algrenewcommand\algorithmicrequire{\textbf{Input:}}
\algrenewcommand\algorithmicensure{\textbf{Output:}}

\title{Splitting Guarantees for Prophet Inequalities via Nonlinear Systems
\vspace{.4cm}
}

\author{
Johannes Brustle
\thanks{Department of Computer and Systems Sciences, Sapienza University of Rome, Italy.}
\and Sebastian Perez-Salazar
\thanks{Department of Computational Applied Mathematics and Operations Research, Rice University, USA.}
\thanks{Ken Kennedy Institute, Rice University, USA.}
\and Victor Verdugo
\thanks{Institute for Mathematical and Computational Engineering, Pontificia Universidad Católica de Chile, Chile.}
\thanks{Department of Industrial and Systems Engineering, Pontificia Universidad Católica de Chile, Chile.}
}

\date{}
\maketitle

\begin{abstract}
The prophet inequality is one of the cornerstone problems in optimal stopping theory and has become a crucial tool for designing sequential algorithms in Bayesian settings. 
In the i.i.d. $k$-selection prophet inequality problem, we sequentially observe $n$ non-negative random values sampled from a known distribution. Each time, a decision is made to accept or reject the value, and under the constraint of accepting at most $k$ items. 
For $k=1$, Hill and Kertz [Ann. Probab. 1982] provided an upper bound on the worst-case approximation ratio that was later matched by an algorithm of Correa et al. [Math. Oper. Res. 2021].
The worst-case tight approximation ratio for $k=1$ is computed by studying a differential equation that naturally appears when analyzing the optimal dynamic programming policy.
A similar result for $k>1$ has remained elusive.

{
In this work, we introduce a nonlinear system of differential equations for the i.i.d. $k$-selection prophet inequality that generalizes Hill and Kertz's equation when $k=1$. 
Our nonlinear system is defined by $k$ constants that determine its functional structure, and their summation provides a lower bound on the optimal policy's asymptotic approximation ratio for the i.i.d. $k$-selection prophet inequality.
To obtain this result, we introduce for every $k$ an infinite-dimensional linear programming formulation that fully characterizes the worst-case tight approximation ratio of the $k$-selection prophet inequality problem for every $n$, and then we follow a dual-fitting approach to link with our nonlinear system for sufficiently large values of $n$.
As a corollary, we use our provable lower bounds to establish a tight approximation ratio for the stochastic sequential assignment problem in the i.i.d. non-negative regime.}
\end{abstract}

\thispagestyle{empty}
\newpage
\setcounter{page}{2}

\section{Introduction}
The prophet inequality problem is one of the cornerstone problems in optimal stopping theory~\citep{krengel1977semiamarts}. 
In the i.i.d. version of the problem, introduced by~\cite{HillKertzStopRule}, a sequence of $n$ i.i.d.\ non-negative values $X_1,\ldots,X_n$ are presented one by one to a decision-maker. 
At each time, the decision-maker faces the choice of either selecting the value or rejecting it entirely, moving on to observe the next value if available, with no option to reconsider previously rejected values. 
The quality of the policy, or algorithm, implemented by the decision-maker is measured by means of the approximation ratio with respect to the expected value of the optimal hindsight (offline benchmark) solution, that is, $\EE[\max\{X_1,\ldots,X_n\}]$. 

~\cite{HillKertzStopRule} provided an algorithm that guarantees an approximation ratio of $1-1/e$ and an upper bound of $\gamma \approx 0.745$ on the approximation ratio by studying the optimal dynamic program for the worst-case distributions. 
Later,~\cite{kertz1986stop} used the recursion from the optimal dynamic program in~\citep{HillKertzStopRule} to provide an ordinary differential equation (ODE)---that we termed Hill and Kertz equation for simplicity and in honor to both authors---where the $\gamma$ bound is embedded as a \emph{unique constant} that guarantees crucial analytical properties of the solution of the ODE: $y'= y(\ln y -1) -1/\gamma + 1$, $y(0)=1$, $y(1)=0$. However, the lower bound on the approximation remained $1-1/e$ for many years until~\cite{correa2021posted} used the Hill and Kertz equation to provide an algorithm that attains an approximation ratio of at least $\gamma$ for any $n$.

In recent years, there has been substantial progress in understanding the approximation limits for prophet inequality problems, mainly driven by their applicability in mechanism design~\citep{lucier2017economic}. One of the most prominent settings is the i.i.d.\ $k$-selection prophet inequality problem, where the decision-maker selects at most $k$ values from the $n$ observed and aims to maximize the expected sum of values selected. 
The offline benchmark in this case is $ \sum_{t=n-k+1}^n \EE[X_{(i)} ]$ where $X_{(1)}\leq X_{(2)}\leq \cdots \leq X_{(n)}$ are the ordered statistics of the random values $X_1,\ldots,X_n$.
Observe that when $k=1$, this setting corresponds to the classic i.i.d. prophet inequality problem.
We refer to this problem as the \emph{$k$-selection prophet inequality}. When the length of 
the sequence is $n$ and $k$ selections can be made, we refer for short to this problem as $\kipi{}$.

The approximation ratio of the i.i.d.\ $k$-selection prophet inequality problem has been proven to be at least $1-k^ke^{-k}/k! \approx 1-1/\sqrt{2\pi k}$ (see, e.g., ~\citep{chakraborty2010approximation,yan2011mechanism,Duetting2020,beyhaghi2021improved,arnosti2023tight}). 
Using a different approach, \cite{jiang2022tightness} recently introduced an optimization framework to characterize worst-case approximation ratios for prophet inequality problems, including the i.i.d.\ $k$-selection setting for a fixed $n$; however, it is unclear how to use their framework to obtain provable worst-case lower bounds for $k\ge 2$.

Over the years, it has remained elusive to get a result analogous to the Hill and Kertz equation for $k\geq 2$, that is, to obtain provable lower bounds on the approximation ratios via studying a closed-form differential equation related to the optimal dynamic programming solution. 
This motivates the central question of this work:
\emph{Can we find a closed-form system of differential equations that provides approximation ratios for the i.i.d.\ $\kipi{}$?} Our findings provide the first step towards provable approximations ratios for $\kipi{}$ via a system of differential equations that naturally extends the Hill and Kertz equation for any $k\geq 1$.


\subsection{Our Contributions and Techniques}

Our first main result characterizes the optimal approximation ratio via an infinite-dimensional linear program in the space of quantiles. In our second main result, we provide a closed-form nonlinear system of differential equations that gives provable lower bounds on the approximation ratio for \kipi{} when $n$ is large enough. In our third result, applying our new provable lower bounds for $\kipi{}$, we find a tight approximation ratio for the stochastic assignment problem. Below, we present more details about our results.

\noindent{\bf Exact Formulation for $\kipi{}$.} {Our first step toward characterizing the asymptotic approximation ratio is an alternative formulation to that of~\cite{jiang2022tightness}. While their formulation provides an exact expression for the approximation ratio of $\kipi{}$ for any $n$ and $k$, its analysis is nontrivial, and even its asymptotic behavior becomes difficult to study beyond the case $k=1$. To address this, we introduce a new infinite-dimensional linear program that \emph{characterizes} the optimal approximation ratio for $\kipi{}$, extending the formulation of~\cite{perez2022iid} to multiple selections. This formulation is inspired by writing the optimal dynamic programming formulation in the quantile space. We take a minimax approach where we search the worst-case distribution while optimizing for the dynamic program's value. We show that in the continuous space of quantiles $[0,1]$, such a problem is linear; see formulation~\ref{form:LP_dual} in Section~\ref{sec:exact_LP_formulation} for the details of the formulation. Even though our formulation is equivalent to the one of~\cite{jiang2022tightness}, using our formulation, we can provide an analysis as $n$ grows that organically provides the nonlinear system of differential equations that we later use to get provable lower bounds. We provide further discussion about the differences between the formulation in~\citep{jiang2022tightness} and ours in Section~\ref{sec:final_remarks}.} In Section~\ref{sec:exact_LP_formulation}, we provide the exact linear programming formulation~\ref{form:LP_dual} and the proof that characterizes the optimal approximation ratio for $\kipi{}$.

\noindent{\bf Approximation via a Nonlinear System.} The analysis of our infinite-dimensional program as $n$ approaches infinity leads us to introduce a system of $k$ coupled nonlinear differential equations, extending the Hill and Kertz equation ($k=1$). 
This new nonlinear system is parameterized by $k$ nonnegative values $\theta_1,\ldots,\theta_k$, and we look for functions $y_1,\ldots,y_k$ satisfying the following in the interval $[0,1)$:
\begin{align}
    (\Gamma_k(-\ln y_k))' & =  k! \left( 1 - 1/(k\theta_k) \right) - \Gamma_{k+1}(-\ln y_k),\label{ode1}\\
    (\Gamma_{k}(-\ln y_j))' & =  k!-\Gamma_{k+1}(-\ln y_j) - \frac{\theta_{j+1}}{\theta_j}(k!-\Gamma_{k+1}(-\ln y_{j+1}))\text{ for every } j\in [k-1],\label{ode2}\\
    y_j(0) & =1 \text{ and }\lim_{t\uparrow 1}y_j(t) =0  \hspace{.4cm}\text{for every } j\in [k],\label{ode3}
\end{align}
where $\Gamma_{\ell}(x)=\int_x^\infty t^{\ell-1} e^{-t}\, \mathrm{d}t$ is the upper incomplete gamma function. {Note that for $k=1$ in~\eqref{ode1}-\eqref{ode3} we recover the Hill and Kertz differential equation.} Now, if $\gamma_{n,k}$ denotes the optimal approximation ratio for \kipi, we can prove that there is $n_0=n_0(k)$ such that for $n\geq n_0$ we have 
\begin{align}
\gamma_{n,k}\geq \left(1- 24 k \frac{\ln(n)^2}{n}\right)\sum_{j=1}^k \theta_j^{\star}, \label{ineq:approximation_informal_result}
\end{align}
where $\theta_1^{\star}, \ldots, \theta_k^\star$ are the values for which there exists a solution to the nonlinear system of differential equations \eqref{ode1}-\eqref{ode3}. We remark that the lower bound on the approximation ratio given by the nonlinear system in inequality~\eqref{ineq:approximation_informal_result} is simply obtained by the summation of the $k$ constants that define it. For instance, when $k=2$, the two constants for which the nonlinear system has a solution are $\theta_1^{\star}\approx 0.346$ and $\theta_2^{\star}\approx 0.483$, and therefore we get a provable lower bound of $\approx 0.829$ on $\gamma_{n,2}$, the optimal approximation ratio for $2$ selections, for any $n$ large enough.

To prove inequality~\eqref{ineq:approximation_informal_result} we employ a dual-fitting approach within our infinite-dimensional linear program. Namely, we introduce a dual infinite-dimensional program of our exact formulation, and using the solution of the nonlinear system~\eqref{ode1}-\eqref{ode3}, we explicitly construct feasible solutions for this dual.
Then, the theorem is obtained by a weak-duality argument. The details are presented in Section~\ref{sec:LB}. We remark that our analysis requires a careful study of the nonlinear system \eqref{ode1}-\eqref{ode3}; which we also provide in Section \ref{sec:LB}. The small multiplicative loss $(1-24k \ln(n)^2/n)$ appears when we construct the dual feasible solution, and it's needed in our analysis to ensure feasibility in the dual problem; note that this loss vanishes as $n$ grows.
We finally note that any feasible solution to our dual program can be implemented using a quantile-based algorithm; therefore, we can implement an algorithm that has an approximation ratio at least $(1-24k\ln(n)^2/n)\sum_{i=1}^k \theta_i^\star$.

\begin{table}[h!]
        \centering
        \begin{tabular}{|c|c|c|c|c|c|}
        \hline
            $k$ & 1 & 2 & 3 & 4 & 5  \\
            \hline\hline
            Our approach ($\sum_{i=1}^k \theta_i^\star$) ($n\to \infty$) & 0.7454 & 0.8293 & 0.8648 & 0.8875 & 0.9035  \\
            \hline
            \cite{beyhaghi2021improved} (any $n\geq 1$)  & 0.6543 & 0.7427 & 0.7857 & 0.8125 & 0.8311 \\
            \hline
            \cite{jiang2022tightness} ($n = 1000$) & 0.7445 & 0.8347 & 0.8718 & 0.8925 & 0.9063 \\
            \hline
        \end{tabular}
        \caption{{  Comparison of known provable lower bounds for $\gamma_{n,k}$ when $n$ is large and $k\in \{1,\ldots,5\}$. The bounds of ~\cite{beyhaghi2021improved} hold for every $n$.} The bounds of ~\cite{jiang2022tightness} are numerically-verified lower bounds on $\gamma_{n,k}$ for $n = 1000$.}
        \label{tab:table_1}
    \end{table}
    
{\noindent\bf Application to the Stochastic Assignment Problem.} Finally, in Section \ref{sec:ssap}, as an application of our new provable lower bounds for $\kipi{}$ we provide the tight optimal approximation ratio for the classic stochastic sequential assignment problem (\ssap{} for short) by~\cite{derman1972sequential}. In the \ssap, there are $n$ non-negative values $r_1,\ldots,r_n$ and $n$ i.i.d.\ non-negative values $X_1,\ldots,X_n$ that are observed one at the time by a decision-maker. At each time $t$, the decision-maker must assign irrevocably $X_t$ to one of the remaining available $r_i$ values that have not been assigned yet. Assigning $X_t$ to $r_i$ provides a reward of $r_i\cdot X_t$, and the goal is to maximize the expected sum of rewards. This problem extends $\kipi{}$ and relates to several online matching problems~\citep{mehta2007adwords,goyal2019online}. 

We revisit the $\ssap{}$ through the lens of prophet inequalities and provide an exact value of its asymptotic approximation ratio. Specifically, we first characterize the optimal approximation ratio for \ssap{} to be equal to $\alpha_n=\min_{k\in [n]}\gamma_{n,k}$. This immediately implies that $\limsup_n \alpha_n \leq \gamma \approx 0.745$. 
To the best of the authors' knowledge, the best current provable lower bounds over $\gamma_{n,k}$ are the following: (1) for $k=1$, $\gamma_{n,1} \geq  \gamma \approx 0.745$~\citep{correa2021posted}; (2) for any $k\geq 1$, $\gamma_{n,k}\geq 1-k^k e^{-k}/k!$ (see, e.g., \citep{Duetting2020,beyhaghi2021improved}); (3) the values reported by~\cite{beyhaghi2021improved} in Table~\ref{tab:table_1}. These three results together imply in principle that $\alpha_n\geq 0.7427$; hence, there is a constant gap between the lower and the upper bound on $\alpha_n$. 
    The $1-k^ke^{-k}/k!$ lower bound for the $\kipi{}$ is at least 0.78 for $k\geq 3$ which is in particular larger than $\liminf_n\gamma_{n,1}\approx 0.745$. 
Nevertheless, to the best of the authors' knowledge, no monotonicity in $k$ is known for the values $\gamma_{n,k}$. Therefore, our new provable $0.829$ lower bound for $k=2$ allows us to conclude that the approximation ratio for the \ssap{} is exactly $\gamma \approx 0.745$ for $n$ sufficiently large, fully characterizing the approximation ratio of the problem.

\subsection{Related Work}

The basic prophet inequality problem, as introduced by \citet{krengel1977semiamarts}, was resolved by using a dynamic program that gave a tight approximation ratio of $1/2$. 
\citet{SamuelCahn1984} later showed that a simple threshold algorithm yields the same guarantee. 
Since then, there have been several generalizations spanning combinatorial constraints,  different valuation functions and arrival orders, resource augmentation, and limited knowledge of the distributions \citep{Kleinberg2012,ehsani2018prophet,Correa2018,CorreaSurvey}. 

A major reason for the renewed interest in prophet inequalities is their relevance to auctions, specifically posted priced mechanisms (PPM) in online sales~\citep{Alaei2011,Chawla2010,Duetting2020,Hajiaghayi2007,Kleinberg2012,Correa2023combinatorial}. It was implicitly shown by \cite{Chawla2010} and \cite{Hajiaghayi2007} that every prophet-type inequality implies a corresponding approximation guarantee in a PPM, and the converse holds as well \citep{Correa2019PricingtoProphets}. 
Using these well-known reductions, our lower bounds for the i.i.d. $k$-selection prophet inequality problem also yield PPM's for the problem of selling $k$ homogeneous goods to $n$ unit-demand buyers who arrive sequentially with independent and identically distributed valuations.

Linear and convex programming have been a powerful tool for the design of online algorithms. For instance, in online and Bayesian matching problems~\citep{mehta2007adwords,goyal2019online}, online knapsack~\citep{babaioff2007knapsack,kesselheim2014primal}, secretary problem~\citep{buchbinder2014secretary,chan2014revealing,perez2021robust}, factor-revealing linear programs~\citep{feldman2016online,lee2018optimal}, and competition complexity~\citep{brustle2022competition}. 
Similar to us, ~\citet{perez2022iid} use a quantile-based linear programming formulation to provide optimal policies in the context of decision-makers with a limited number of actions.

Our analysis provides a new nonlinear system of differential equations, which extends the ordinary differential equation by Hill and Kertz for $k=1$, and provides provable lower bounds on the asymptotic approximation ratio.
The standard Hill and Kertz equation has been used in various recent works~\citep{correa2021posted, Liu2021, perez2022iid, brustle2022competition} to provide guarantees in single-selection problems. Likewise, ODE methods have been used in other online selection problems~\citep{assaf2002ratio,ekbatani2024prophet,mucci1973differential,perez2021robust} to provide asymptotic guarantees. {Recently, the concurrent work by~\citep{molina2025prophet} derived a nonlinear system similar to~\eqref{ode1}-\eqref{ode3} for a variation of $\kipi{}$. They obtained this system through a nontrivial extension of the quantile approach introduced by~\cite{correa2021posted}. Our approach differs from theirs in that we take an optimization-based perspective, similar to~\citep{jiang2022tightness,perez2022iid}.}

\section{Preliminaries}

An instance of $\kipi{}$ is given by a tuple $(n,k,F)$, where $n$ is the number of values $X_1, \ldots, X_n$ that are drawn i.i.d according to the continuous distribution $F$ supported on $\RR_+$. This assumption is commonly made in the literature (e.g., \citet{Liu2021}), as we can perturb a discrete distribution by introducing random noise at the cost of a negligible loss in the objective. 
Given an instance of \kipi{}, observe that we can always scale the values $X_1,\ldots,X_n$ by a positive factor so the optimal value is equal to 1, and the reward of the optimal policy is scaled by the same amount.
In particular, the approximation ratio of the optimal policy remains the same.

Given an instance $(n,k,F)$, we use dynamic programming to compute the optimal reward of the optimal sequential policy. 
Let $A_{t,\ell}(F)$ be the reward of the optimal policy when $\ell\leq k$ choices are still to be made in periods $\{t,\ldots,n\}$.
Then, for every $t\in [n]$ and $\ell \in [k]$, the following holds:
\begin{align}
    A_{t,\ell}(F) &= \max_{x\geq 0} \Big\{  (\EE[X\mid X\geq x]  + A_{t+1,\ell-1}(F))\Pr[X\geq x] + A_{t+1,\ell}(F)\Pr[X<x]  \Big\}, \label{eq:DP_recursion} \\
    A_{n+1,\ell}(F) & = 0,\text{ and }   
    A_{t,0}(F)  = 0.\hspace{.3cm} \label{eq:DP_init_condi_2}
\end{align}
Equation~\eqref{eq:DP_recursion} corresponds to the continuation value condition in optimality; the term in the braces is the expected value obtained when a threshold $x$ is chosen when at period $t$ and $\ell$ choices can still be made. In~\eqref{eq:DP_init_condi_2} we have the border conditions. 
In particular, it holds 
\begin{align}\label{eq:gamma-def}
\gamma_{n,k} = \inf\Big\{A_{1,k}(F): \text{instances }(n,k,F)\text{ with }\OPT_{n,k}(F)=1\Big\}. 
\end{align}

\section{An Infinite-Dimensional Formulation}\label{sec:exact_LP_formulation}

In this section, we provide the characterization of the optimal approximation ratio for $\kipi{}$, $\gamma_{n,k}$, via an infinite-dimensional linear program. For every positive integers $k$ and $n$, with $n\ge k$, consider the following infinite dimensional linear program:
\begin{align}
	\inf \quad\,\,  d_{1,k}\hspace{2.5cm}& \tag*{\normalfont{\mbox{[P]$_{n,k}$}}}\label{form:LP_dual}\\
	\normalfont\text{s.t.} \quad d_{t,\ell} \geq \int_0^q h(u) \, \mathrm{d}u \, +\,  & q d_{t+1,\ell-1} + (1-q)d_{t+1,\ell},  \text{ for every } t\in[n], \ell\in[k],\text{ and } q\in [0,1], \label{const:dynamic_constr_dual} \\
	\int_0^{1} g_{n,k}(u) h(u) \, \mathrm{d}u &\geq 1,\label{const:max_value_const_dual} \\
	h(u) &\geq h(w)\hspace{.2cm} \text{ for every } u\leq w, \text{with }u,w\in [0,1],\label{const:nonincreasing_dual}\\
    d_{t,\ell}&\ge 0\hspace{0.3cm} \text{for every }t\in [n+1]\text{ and every }\ell\in \{0\}\cup [k],\label{const:d-positive-dual}\\
    h(u)&\ge 0\hspace{.2cm}\text{ for every }u\in [0,1]\text{ and }h\text{ is continuous in }(0,1),\label{const:positive-f-dual}
\end{align}
where $g_{n,k}(u)=\sum_{j=n-k+1}^n j\binom{n}{j}(1-u)^{j-1} u^{n-j}$ for every $u\in [0,1]$. $d_{t,\ell}$ represents that the welfare of the optimal dynamic policy after having seen $t\leq n$ values and $\ell \leq k$ of them remain to be selected. The family of constraints (\ref{const:dynamic_constr_dual}) asserts that for any quantile $q$, $d_{t,\ell}$ is at least the welfare obtained by using $q$ as a threshold for selection in the current round, $t$. That is, the first two terms of the right hand side cover the case when the value is indeed selected, whereas the third term covers the case when it is not. Constraint (\ref{const:max_value_const_dual}) represents a normalization to instances where the optimal welfare $ \sum_{t=n-k+1}^n \EE[X_{(i)} ]\geq 1$. 
In the program \ref{form:LP_dual}, the variables $h(u)$ represent the values of a non-negative and non-increasing function $h$ in $[0,1]$, and therefore we have infinitely many of them, whereas the variables $d$ are finitely many. More specifically, $h(u)$ represents $F^{-1}(1-u)$, where $F$ is the c.d.f of a probability distribution and minimizing over all possible choices of $h$ is equivalent to search the worst-case distribution for $\kipi{}$.
We remark the continuity for $h$ in the program is mainly for the sake of simplicity in our analysis but does not represent a strict requirement.

The following structural result formalizes the interpretation of the variables provided in the previous paragraph. The remainder of the section is dedicated to its proof.
\begin{theorem}\label{thm:new-LP}
The optimal approximation ratio for \kipi{} is equal to the optimal value of \ref{form:LP_dual}.
\end{theorem}

We denote by $v_{n,k}$ the optimal value of~\ref{form:LP_dual}. 
We show that $v_{n,k}=\gamma_{n,k}$ in Theorem \ref{thm:new-LP} by proving both inequalities, $v_{n,k}\leq \gamma_{n,k}$ and $v_{n,k}\geq \gamma_{n,k}$, separately. 
For the first inequality, we argue that any instance of \kipi{} with $\OPT_{n,k}(F)=1$, produces a feasible solution to~\ref{form:LP_dual} with an objective value equal to the reward of the optimal sequential policy. 
For the second, we show that any feasible solution $(d,f)$ to~\ref{form:LP_dual} produces an instance of \kipi{} such that the reward of the optimal policy is no larger than the objective value of the instance $(d,f)$.

Before we prove the inequalities, we leave a proposition with some preliminary properties that we use in our analysis. 
The proof can be found in Appendix \ref{app:exact_LP_formulation}.
\begin{proposition}\label{prop:useful_LP}
Let $F$ be a continuous and strictly increasing distribution over the non-negative reals.
Then, the following properties hold:
\begin{enumerate}[itemsep=0pt,label=\normalfont(\roman*)]
    \item For every $n$ and $k$ with $n\ge k$, we have $\OPT_{n,k}(F) = \int_0^1 g_{n,k}(u) F^{-1}(1-u)\, \mathrm{d}u$.\label{prop:change_variable_opt}
    \item Suppose that $X$ is a random variable distributed according to $F$.
    Then, for every $x\ge 0$, it holds $\EE[X\mid X\geq x]\Pr[X \geq x] = \int_0^q F^{-1}(1-u)\, \mathrm{d}u$, where $q=\Pr[X\geq x]$.\label{prop:change_variable_simple_exp}
\end{enumerate}
\end{proposition}
We use the following two lemmas to prove Theorem \ref{thm:new-LP}.
\begin{lemma}\label{lem:LP-exact-1}
    Let $F$ be a continuous and strictly increasing distribution over the non-negative reals such that $\OPT_{n,k}(F)=1$, and let $h(u)=F^{-1}(1-u)$ for every $u\in [0,1]$.
    Then, $(A(F),h)$ is feasible for \ref{form:LP_dual}, where $A(F)=(A_{t,\ell}(F))_{t,\ell}$ is defined according to \eqref{eq:DP_recursion}-\eqref{eq:DP_init_condi_2}.
\end{lemma}

\begin{proof}
    By construction, we have $h\geq 0$ and $h$ is non-increasing, therefore constraints \eqref{const:nonincreasing_dual}-\eqref{const:positive-f-dual} are satisfied by $(A(F),h)$. 
    Furthermore, we have
    $$\int_0^1 g_{n,k}(u) h(u)\, \mathrm{d}u=\int_0^1 g_{n,k}(u) F^{-1}(1-u)\, \mathrm{d}u = \OPT_{n,k}(F) = 1,$$
    where the second equality holds by Proposition~\ref{prop:useful_LP}\ref{prop:change_variable_opt}.
    Then, constraint~\eqref{const:max_value_const_dual} is satisfied by $(A(F),h)$.
    Let $q\in [0,1]$, and $x\geq 0$ such that $q=\Pr[X\geq x]$. Then, for every $t\in [n]$ and every $\ell\in [k]$, we have
    \begin{align*}
        A_{t,\ell}(F) & \geq (\EE[X\mid X\geq x] + A_{t+1,\ell-1}(F))\Pr[X\geq x] + A_{t+1,\ell}(F)\Pr[X<x] \\
        & = \int_0^q h(u)\, \mathrm{d}u + q A_{t+1,\ell-1}(F)+ (1-q)A_{t+1,\ell}(F),
    \end{align*}
    where in the first inequality we used condition \eqref{eq:DP_recursion}, while in the second inequality, we used that $q=\Pr[X\geq x]$ and Proposition \ref{prop:useful_LP}\ref{prop:change_variable_simple_exp}. 
    Then, $(A(F),h)$ satisfies constraint \eqref{const:dynamic_constr_dual}, and we conclude that $(A(F),h)$ is feasible for \ref{form:LP_dual}.
\end{proof}

\begin{lemma}\label{lemm:LP-exact-2}
    Let $(d,h)$ be any feasible solution for \ref{form:LP_dual}.
    Then, there exists a probability distribution $G$, such that $d_{t,\ell}\geq A_{t,\ell}(G)$ for every $t\in [n]$ and $\ell\in [k]$.
\end{lemma}
\begin{proof}
    Given a feasible solution $(d,h)$ for \ref{form:LP_dual}, 
    consider the random variable $h(1-Q)$, where $Q$ is a uniform random variable over the interval $[0,1]$, and let $G$ be the probability distribution of $h(1-Q)$. 
Then, $G$ is continuous and non-decreasing.
Since $h$ is a non-increasing function, we have
$\Pr[h(1-Q) \leq h(u)] \geq \Pr[1-Q \geq u] = \Pr[1-u \geq Q] = 1-u,$ and therefore $G^{-1}(1-u)\le h(u)$.
    We prove the inequalities in the lemma statement for $d$ and the probability distribution $G$ via backward induction in $t\in \{1,\ldots,n+1\}$. 
    For $t=n+1$ we have $d_{n+1,\ell}=0=A_{n+1,\ell}(G)$ for any $\ell\in [k]$. 
    Assume the result holds true for $\{t+1,\ldots,n+1\}$.
    If $\ell=0$, we have $d_{t,\ell}=0=A_{t,\ell}(G)$, and for $\ell\in [k]$, we have
\begin{align*}
    d_{t,\ell} & \geq \sup_{q\in [0,1]} \left\{ \int_0^q h(u)\, \mathrm{d}u + q d_{t+1,\ell-1} + (1-q)d_{t+1,\ell}  \right\} \\
    & \geq \sup_{q\in [0,1]} \left\{  \int_0^q h(u) \, \mathrm{d}u + q A_{t+1,\ell-1}(G) + (1-q)A_{t+1,\ell}(G)  \right\}\\
    & \geq \sup_{q\in [0,1]} \left\{  \int_0^q G^{-1}(1-u) \, \mathrm{d}u + q A_{t+1,\ell-1}(G) + (1-q)A_{t+1,\ell}(G)  \right\}\\
    & = \max_{x\geq 0} \left\{  (\EE[X\mid X\geq x]  + A_{t+1,\ell-1}(G))\Pr[X\geq x] + A_{t+1,\ell}(G)\Pr[X<x]  \right\}\\
    & = A_{t,\ell}(G),
\end{align*}
where the first inequality holds since $(d,h)$ satisfies constraint~\eqref{const:dynamic_constr_dual}; the second inequality holds by induction; the third holds since $G^{-1}(1-u)\le h(u)$, and the first equality by Proposition \ref{prop:useful_LP}\ref{prop:change_variable_simple_exp}.
This concludes the proof of the lemma.
\end{proof}

\begin{proof}[Proof of Theorem \ref{thm:new-LP}]
By Lemma \ref{lem:LP-exact-1}, for every probability distribution $F$ we have that $(A(F),h)$ is feasible for \ref{form:LP_dual}, and its objective value is equal to $A_{1,k}(F)$.
This implies that $A_{1,k}(F)\ge v_{n,k}$ for every $F$, and therefore $\gamma_{n,k}\ge v_{n,k}$.
By Lemma \ref{lemm:LP-exact-2}, for every feasible solution $(d,h)$ in \ref{form:LP_dual} there exists a probability distribution $G$ such that $d_{1,k}\ge A_{1,k}(G)$, and therefore, the optimal value of \ref{form:LP_dual} is lower bounded by the infimum in \eqref{eq:gamma-def}, which is equal to $\gamma_{n,k}$.
We conclude that $v_{n,k}\ge \gamma_{n,k}$, and therefore both values are equal.
\end{proof}

\section{Lower Bound on the Approximation Ratio}\label{sec:LB}

In this section, we prove the following result.

\begin{theorem}\label{thm:comp-ratio}
For every $k\geq 1$, there exists $n_0\in \NN$, such that for every $n\ge n_0$ we have 
$$\gamma_{n,k}\geq \left(1- 24 k \frac{\ln(n)^2}{n}\right)\sum_{j=1}^k \theta_j^{\star},$$ where $\theta_1^{\star}, \ldots, \theta_k^\star$ are the values for which there exists a solution to the nonlinear system of differential equations \eqref{ode1}-\eqref{ode3}.
\end{theorem}

Before proving this theorem, we provide a warm-up for the case of $k=2$, i.e., when 2 selections are possible. The goal of Section \ref{subsec:warm-up} is to provide the main insights into the derivation of the nonlinear system~\eqref{ode1}-\eqref{ode3} and to sktech the main steps in the proof of Theorem~\ref{thm:comp-ratio}. In Section~\ref{subsec:proof_of_lower_bound_asympt} we provide the full detailed proof of Theorem~\ref{thm:comp-ratio} which holds for every $k$.

\subsection{Warm-Up: The case of $k=2$}\label{subsec:warm-up}

In this subsection, we sketch the deduction of the nonlinear system~\eqref{ode1}-\eqref{ode3}. We focus on the case $k=2$ and provide a weak dual formulation of~\ref{form:LP_dual}. From this weak dual, we can find a 
recursion that in the limit converges to a solution of~\eqref{ode1}-\eqref{ode3} with $k=2$. While this approach only gives us a lower bound on $\liminf_{n} \gamma_{n,2}$, we strengthen this asymptotic result by showing how to transform a solution of the system~\eqref{ode1}-\eqref{ode3} into a feasible solution to our dual LP by incurring in a slight loss when $n$ is large.

\paragraph{A Dual for $k=2$.} Consider the following infinite-dimensional linear program:
\begin{align}
    \sup \quad  v \hspace{1cm} &\tag*{\normalfont{\mbox{[D]$_{n,2}$}}}\label{form:dual-relaxation_k_2}\\
	\normalfont\text{s.t.} \quad  \int_0^1 \beta_{1,\ell}(q)\, \mathrm{d}q &\leq \mathbf{1}_{2}(\ell),  \hspace{.3cm} \text{for all }\ell\in [2],\label{const:constr_1_P_n_2-relax} \\
	\int_0^1 \beta_{t+1,2}(q) \, \mathrm{d}q&\leq \int_0^1 (1-q) \beta_{t,2}(q) \, \mathrm{d}q,\hspace{.3cm}   \text{for all } t\in [n-1] ,\label{const:constr_t_P_n_2-relax}\\
	\int_0^1 \beta_{t+1,1}(q) \, \mathrm{d}q&\leq \int_0^1 (1-q) \beta_{t,1}(q) \, \mathrm{d}q+ \int_0^1 q \beta_{t,2}(q) \, \mathrm{d}q,   \hspace{.2cm}\text{for all } t\in [n-1],\label{const:constr_t_P_n_22-relax}\\
	 v g_{n,2}(u)&\leq \sum_{t=1}^n \int_u^1 \beta_{t,1}(q) + \beta_{t,2}(q) \, \mathrm{d}q,  \text{ for }u \in [0,1]  , \label{const:P_n_approx-relax_2} \\
	\beta_{t,\ell}(q)&\ge 0\hspace{.3cm} \text{ for all }q\in [0,1],t\in [n]\text{ and }\ell\in [2],\label{const:P_n_initial-relax_2}
\end{align}
where $\mathbf{1}_2(\ell)\in \{0,1\}$ and $\mathbf{1}_2(\ell) =1$ if and only if $\ell=2$. This linear program can be interpreted as follows. Fix an algorithm that makes decisions based on quantiles. Now, the variables $\beta_{t,\ell}(q)$ can be interpreted as the probability densities of the events that \emph{the algorithm chooses quantile $q$ for time $t$ when $\ell$ items remain to be chosen, observes the $t$-th value $x_t$, and selects it if $x_t \geq F^{-1}(1-q)$}. Variable $v\geq 0$ captures the approximation ratio of such algorithm. Constraints~\eqref{const:constr_1_P_n_2-relax}-\eqref{const:constr_t_P_n_22-relax} capture the valid transitions from time $t$ to $t+1$. Specifically, for the algorithm to observe a value at time $t$ when $\ell$ items can still be chosen, it must have observed a value at time $t-1$ under one of the following conditions: Either (1) there are $\ell$ items that can be chosen at $t-1$ but the algorithm did not select the $t-1$ observed value; or (2) there are $\ell+1$ items that can be chosen at $t-1$ and the algorithm selected the $t-1$ observed value. Constraint~\eqref{const:P_n_approx-relax_2} relates the offline density of least two out of $n$ values being in the top $u$ quantiles with the density that the algorithm selects values in the same quantile. Among all possible quantiles $u\in [0,1]$, the largest $v$ such that the ratio of these densities is at least $v$ corresponds to a lower bound on the approximation ratio of the algorithm. 

In Lemma~\ref{lem:weak_duality}, we show that~\ref{form:dual-relaxation_k_2} is a weak dual to the exact formulation~\ref{form:LP_dual}, for $k=2$. That is, the optimal value $w_{n,2}$ of~\ref{form:dual-relaxation_k_2} is at most $\gamma_{n,2}$, which is the optimal approximation ratio for $2$ selections. Thus, finding solutions to the weak dual provide a mechanism to give provable lower bounds on $\gamma_{n,2}$. We remark that~\ref{form:dual-relaxation_k_2} is not necessarily a strong dual to~\ref{form:LP_dual}. Variables $\beta_{t,\ell}$ correspond to constraints~\eqref{const:dynamic_constr_dual} and variable $v$ corresponds to constraint~\eqref{const:max_value_const_dual}; however, in~\ref{form:dual-relaxation_k_2} there are no dual variables for constraints~\eqref{const:nonincreasing_dual}. 

\paragraph{A Feasible Solution and Its Limit.} We now construct a particular feasible solution to the weak dual~\ref{form:dual-relaxation_k_2} and show that it produces a set of points that converge to a solution of~\eqref{ode1}-\eqref{ode3} for $k=2$. Inspired by the quantile-based solution of~\cite{correa2021posted} and the LP characterization by~\cite{perez2022iid}, we propose the following solution to~\ref{form:dual-relaxation_k_2}:
\begin{align}
\beta_{t,2}(q) = \theta_2 \cdot (-g'_{n,2}(q)) \, \mathbf{1}_{(\varepsilon_{t-1},\varepsilon_t)}(q), \qquad \beta_{t,1}(q) = \theta_1 \cdot (-g'_{n,2}(q)) \, \mathbf{1}_{(\mu_{t-1},\mu_t)}(q), \label{eqs:beta_definition}
\end{align}
for $t=1,\ldots,n$, where $\theta_1,\theta_2\geq 0$ and $0=\varepsilon_0<\varepsilon_1< \cdots <\varepsilon_n$ and $0=\mu_0 = \mu_1 < \mu_2 < \cdots < \mu_n$. Note that if $\varepsilon_n=\mu_n=1$, then $(\beta,\theta_1+\theta_2)$ is a feasible solution to~\ref{form:LP_dual}; hence, $\theta_1+\theta_2\leq \gamma_{n,2}$. Now, assuming that constraints~\eqref{const:constr_1_P_n_2-relax}-\eqref{const:constr_t_P_n_22-relax} are tightened by $\beta$, we can deduce the following implicit recursions for $\varepsilon_t$ and $\mu_t$:
\begin{align}
g_{n,2}(\varepsilon_{t+1}) - g_{n,2}(\varepsilon_t) & = -\frac{1}{\theta_2} - \frac{2}{n+1}\left( g_{n+1,3}(\varepsilon_t) -  (n+1) \right) & t \in \{0\}\cup [n], \label{eq:recursion_for_eps}\\
g_{n,2} (\mu_{t+1}) - g_{n,2}(\mu_t) & = \frac{\theta_2}{\theta_1} \frac{2}{n+1}\left( g_{n+1,3}(\epsilon_t) - (n+1) \right) & t\in \{0\}\cup [n], \label{eq:recursion_for_mu} \\
& \quad  - \frac{2}{n+1}\left( g_{n+1,3}(\mu_t) - (n+1) \right) \nonumber
\end{align}
In other words, by fixing $\theta_2 >0$, we can find a sequence of increasing $\varepsilon_t$'s. Likewise, by fixing $\theta_1,\theta_2>0$, we can find a sequence of increasing $\mu_t$'s. Our goal is to find a pair $(\theta_1^\star,\theta_2^\star)$ such that $\varepsilon_n(\theta_2)=1$ and $\mu_n(\theta_1,\theta_2)=1$.
A simple inductive argument shows that $\varepsilon_t$ is decreasing in $\theta_2$. Indeed, by differentiating~\eqref{eq:recursion_for_eps} in $\theta_2$, we obtain $g_{n,2}'(\varepsilon_{t+1}(\theta_2))\varepsilon_{t+1}'(\theta_2) = {1}/{\theta_2^2} + \varepsilon_{t}'(\theta_2) (1-\varepsilon_{t}(\theta_2))g_{n,2}'(\varepsilon_t(\theta_2))$.
From here, and also using~\eqref{eq:recursion_for_eps}, we can show that: (1) $\varepsilon_t$ blows up when $\theta_1$ goes to $0$ (2) $\varepsilon_t$ tends to $0$ when $\theta_2$ goes to $+\infty$ and (2) there this a unique $\theta_2^\star > 0$ such that $\varepsilon_n(\theta_2^\star)=1$. 
Analogously, having already found and fixed $\theta_2^\star>0 $, we can show that there is a unique $\theta_1^\star>0$ such that $\mu_n(\theta_1^\star,\theta_2^\star)=1$. This shows that our solution $(\beta,\theta_1^\star+\theta_2^\star)$ is feasible to~\ref{form:dual-relaxation_k_2}. 

The pair $(\theta_1^\star, \theta_2^\star)$ depends non-trivially on $n$. Thus, to tackle this dependency, we perform a limit analysis on~\eqref{eq:recursion_for_eps} and~\eqref{eq:recursion_for_mu}. For this, we first perform a change of variable $y_{1,t}=e^{- n \mu_t}$ and $y_{2,t}=e^{- n \varepsilon_t}$ for all $t$. We see that for all $\ell\in [2]$, $y_{\ell,t} \in [0,1]$, $y_{\ell,1}=1$, $y_{\ell,n}=e^{-n}$ for all $\ell\in [2]$. Now, consider the function $\hat{y}_\ell : [0,1]\to [0,1]$ such that, $\hat{y}_\ell(0)=1$, $\hat{y}_\ell(t/n)= y_{\ell,t}$, for all $t\in \{0\}\cup [n]$, and $\hat{y}_\ell$ is linear between in $[(t-1)/n ,t/n]$ for $t\in [n]$. Now using~\eqref{eq:recursion_for_eps} and~\eqref{eq:recursion_for_mu}, the approximation $g_{n,k}(u)\approx  n \Gamma_k(n u)/(k-1)!$,\footnote{in Proposition~\ref{prop:g-n-properties} we prove formally this approximation for the derivative of $g_{n,k}$.} and assuming that $(\theta_1^\star,\theta_2^\star)$ converges when $n\to \infty$, we obtain that $(\hat{y}_1,\hat{y}_2)$ converges to a solution of the nonlinear system~\eqref{ode1}-\eqref{ode3} with $k=2$. A formal proof of the existence of a solution to the nonlinear system~\eqref{ode1}-\eqref{ode3} for general $k$ appears in Lemma~\ref{prop:nls}.

\paragraph{From the Limit to a Feasible Solution.} The limit analysis of the solution $\beta$ only bounds $\liminf_n \gamma_{n,2}$. To provide an analysis for finite $n$ as in Theorem~\ref{thm:comp-ratio}, we construct an explicit solution to~\ref{form:dual-relaxation_k_2} from the solution $(y_1,y_2)$ of~\eqref{ode1}-\eqref{ode3}. While we are tempted to take $\varepsilon_t=-\ln (y_2(t/n))/n$ and $\mu_t=-\ln (y_1(t/n))/n$ and defining $\beta$ as in~\eqref{eqs:beta_definition}, this poses two nontrivial challenges. Firstly, $\varepsilon_n$ and $\mu_n$ are larger than $1$ since $y_1(1)=y_2(1)=0$. Secondly, $\mu_1$ is positive; yet for $\beta$ to be feasible to~\ref{form:dual-relaxation_k_2}, $\mu_1$ must be $0$, due to constraint~\eqref{const:constr_1_P_n_2-relax}.

To address the first challenge, we take $\bar{n}\leq n$, and define $\varepsilon_t=-\ln (y_2(t/\bar{n}))/\bar{n}$ for $t\in \{0\}\cup [\bar{n}-1]$ and $\varepsilon_{\bar{n}}=1$ for some $\bar{n}\leq n$. Similarly, $\mu_0=\mu_1=0$, $\mu_t=-\ln (y_1(t/\bar{n}))/\bar{n}$, for $t\in [\bar{n}-1]\setminus\{1\}$, and $\mu_{\bar{n}}=1$. The value $\bar{n}$ appears when approximating $g_{n,k}'(u)$ with $\bar{n}\Gamma_k({\bar{n}} u)'/(k-1)!$ (see Proposition~\ref{prop:g-n-properties}); for $k=2$ we have $\bar{n}=n-3$. For $n$ large enough, we show that $-\ln( y_{\ell}(1-1/\bar{n}))/\bar{n} < 1$ for $\ell\in [2]$; hence, our new sequence of $\varepsilon$'s and $\mu$'s is well-defined.

For the second challenge, on a first read, it might seem that defining $\mu_1=0$ solves the problem. However, this is not the case because, in constraint~\eqref{const:constr_t_P_n_22-relax} for $t=2$, the left-hand side could be larger than the right-hand side as our choice of $\varepsilon_t$ and $\mu_t$ mimics an Euler approximation of the nonlinear system~\eqref{ode1}-\eqref{ode3}; hence, $\mu_2$ depends on $\mu_1$ in the Euler approximation; however, we defined it to be $0$. To address this, we add enough mass to $\beta_{1,2}$ so constraint~\eqref{const:constr_t_P_n_22-relax} approximately holds. A complete description of the solution appears in~\eqref{eq:description_of_solution}.

In general, these two fixes introduce two sources of multiplicative loss in the objective value $\theta_1^\star+\theta_2^\star$. The first is via $\bar{n}$, and the second is via the mass addition to $\beta_{1,2}$. We show that these multiplicative losses are in the order of {$1-c \ln (n)^2/n$} for some constant $c$, and they vanish as $n$ grows.

\subsection{Main Analysis}\label{subsec:proof_of_lower_bound_asympt}

We now provide the proof of Theorem~\ref{thm:comp-ratio}. Our proof is organized into three main steps. In the first step, we introduce a \emph{maximization} infinite-dimensional linear program that we call \ref{form:dual-relaxation}, and we prove that weak duality holds for the pair \ref{form:LP_dual} and \ref{form:dual-relaxation}.
Namely, the optimal value of \ref{form:dual-relaxation} provides a lower bound on the optimal value of~\ref{form:LP_dual}. The program~\ref{form:dual-relaxation} can be formally deduced from~\ref{form:LP_dual} by dropping Constraint~\eqref{const:nonincreasing_dual}; however, the weak duality still requires a proof as we are dealing with infinite-dimensional programs.

In the second step, we introduce a second maximization infinite-dimensional linear program parametrized by a value { $\bar{n}\leq n$}, namely~\ref{form:P_nk_restricted_n}. 
This program is akin to~\ref{form:dual-relaxation}, but is described by a set of constraints that become more handy when analyzing the nonlinear system. Furthermore,~\ref{form:P_nk_restricted_n} restricts the time horizon until { $\bar{n}$}. 
{ We show that the optimal value of~\ref{form:P_nk_restricted_n} provides a lower bound on the program~\ref{form:dual-relaxation}, and therefore, it gives a lower bound on the optimal value of~\ref{form:LP_dual} as well.}

In the third step, we build an explicit feasible solution to the problem~\ref{form:P_nk_restricted_n} starting from a solution of the nonlinear system of differential equations \eqref{ode1}-\eqref{ode3}. Using the valid bounds found in the previous two steps, we can provide a lower bound on the value $\gamma_{n,k}$ for $n$ large enough. In particular, we show, as $n$ grows, that the sequence of lower bounds provides a lower bound on the optimal asymptotic approximation factor.\\

\noindent{\bf First step: Weak duality.}
For every $\ell\in [k]$, let $\mathbf{1}_{k}(\ell)=1$ if $\ell=k$ and $\mathbf{1}_{k}(\ell)=0$ for $\ell\neq k$.
Consider the following infinite-dimensional linear program:
\begin{align}
    \sup \quad  v \hspace{1cm} &\tag*{\normalfont{\mbox{[D]$_{n,k}$}}}\label{form:dual-relaxation}\\
	\normalfont\text{s.t.} \quad  \int_0^1 \beta_{1,\ell}(q)\, \mathrm{d}q &\leq \mathbf{1}_{k}(\ell),  \hspace{.3cm} \text{for all }\ell\in [k],\label{const:constr_1_P_n_k-relax} \\
	\int_0^1 \beta_{t+1,k}(q) \, \mathrm{d}q&\leq \int_0^1 (1-q) \beta_{t,k}(q) \, \mathrm{d}q,\hspace{.3cm}   \text{for all } t\in [n-1] ,\label{const:constr_t_P_n_k-relax}\\
	\int_0^1 \beta_{t+1,\ell}(q) \, \mathrm{d}q&\leq \int_0^1 (1-q) \beta_{t,\ell}(q) \, \mathrm{d}q+ \int_0^1 q \beta_{t,\ell+1}(q) \, \mathrm{d}q,   \hspace{.2cm}\text{for all } t\in [n-1], \ell\in[k-1],\label{const:constr_t_P_n_2-relax}\\
	 v g_{n,k}(u)&\leq \sum_{t=1}^n\sum_{\ell=1}^k \int_u^1 \beta_{t,\ell}(q) \, \mathrm{d}q,  \text{ for }u \in [0,1]  , \label{const:P_n_approx-relax} \\
	\beta_{t,\ell}(q)&\ge 0\hspace{.3cm} \text{ for all }q\in [0,1],t\in [n]\text{ and }\ell\in [k].\label{const:P_n_initial-relax}
\end{align}
The variables $\beta_{t,\ell}(q)$ represent the probability density of an optimal algorithm choosing quantile $q$ at time $t$ when $\ell$ items can still be chosen, and variable $v$ captures the approximation factor of the policy. (See also the interpretation of the linear program in subsection~\ref{subsec:warm-up}.) We denote by $w_{n,k}$ the optimal value of \ref{form:dual-relaxation}.
The following lemma shows that weak duality holds for the pair of infinite-dimensional programs \ref{form:LP_dual} and \ref{form:dual-relaxation}.


\begin{lemma}\label{lem:weak_duality}
For every $n\geq 1$ and every $k\in \{1,\ldots,n\}$, we have $v_{n,k}\ge w_{n,k}$.
\end{lemma}

\begin{proof}
Consider a feasible solution $(d,h)$ for \ref{form:LP_dual} and a feasible solution $(\beta,v)$ for \ref{form:dual-relaxation}.
Since~\ref{form:LP_dual} is a minimization problem, we can assume that $d_{n+1,\ell}=0$ for every $\ell\in [k]$ and $d_{0,\ell}=0$ for every $t\in [n]$; if they are non-zero, we can easily make them zero without changing the objective value of $(d,h)$.
In what follows, we show that $v\le d_{1,k}$.
Since $(\beta,v)$ satisfies constraint \eqref{const:P_n_approx-relax} and $h(u)\ge 0$ by constraint \eqref{const:positive-f-dual}, we get 
\begin{align}
v \int_{0}^{1}g_{n,k}(u)h(u)\, \mathrm{d}u&\leq \int_{0}^{1}\sum_{t=1}^n\sum_{\ell=1}^k \left(\int_u^1 \beta_{t,\ell}(q) \, \mathrm{d}q\right)h(u)\, \mathrm{d}u\notag\\
&=\sum_{t=1}^n\sum_{\ell=1}^k\int_{0}^{1} \int_u^1 \beta_{t,\ell}(q) h(u)\, \mathrm{d}q\, \mathrm{d}u\notag\\
&=\sum_{t=1}^n\sum_{\ell=1}^k\int_{0}^1 \beta_{t,\ell}(q)\int_0^{q}  h(u)\, \mathrm{d}u\, \mathrm{d}q,\label{ineq:weak-duality-1}
\end{align}
where the first equality holds by exchanging the summation and the integrals, and the second equality holds by exchanging the integration order for $u$ and $q$. Then, inequality \eqref{ineq:weak-duality-1} together with constraint~\eqref{const:max_value_const_dual} imply that
\begin{equation}
v\le \sum_{t=1}^n\sum_{\ell=1}^k\int_{0}^1 \beta_{t,\ell}(q)\int_0^q  h(u)\, \mathrm{d}u\, \mathrm{d}q.\label{ineq:weak-duality-2}
\end{equation}
On the other hand, from constraint \eqref{const:dynamic_constr_dual}, for every $t\in [n]$ and every $\ell\in [k]$ we have
\begin{align}
&\int_{0}^1\beta_{t,\ell}(q)\int_0^q h(u) \, \mathrm{d}u\,\mathrm{d}q\notag\\ 
&\leq d_{t,\ell}\int_{0}^1\beta_{t,\ell}(q)\,\mathrm{d}q -d_{t+1,\ell-1}\int_{0}^1q\beta_{t,\ell}(q)\, \mathrm{d}q  - d_{t+1,\ell}\int_{0}^1(1-q)\beta_{t,\ell}(q)\, \mathrm{d}q,\label{ineq:weak-duality-3}
\end{align}
where we used that $\beta_{t,\ell}\ge 0$, and then we integrated over $q\in [0,1]$.
When $\ell=k$, note that
\begin{align}
&\sum_{t=1}^nd_{t,k}\int_{0}^1\beta_{t,k}(q)\,\mathrm{d}q-\sum_{t=1}^nd_{t+1,k}\int_{0}^1(1-q)\beta_{t,k}(q)\, \mathrm{d}q\notag\\
&=d_{1,k}\int_{0}^1\beta_{1,k}(q)\,\mathrm{d}q+\sum_{t=2}^nd_{t,k}\int_{0}^1\beta_{t,k}(q)\,\mathrm{d}q-\sum_{t=1}^nd_{t+1,k}\int_{0}^1(1-q)\beta_{t,k}(q)\, \mathrm{d}q\notag\\
&\le d_{1,k}+\sum_{t=1}^{n-1}d_{t+1,k}\int_{0}^1\beta_{t+1,k}(q)\,\mathrm{d}q-\sum_{t=1}^nd_{t+1,k}\int_{0}^1(1-q)\beta_{t,k}(q)\, \mathrm{d}q\notag\\
&\le d_{1,k}+\sum_{t=1}^{n-1}d_{t+1,k}\int_{0}^1(1-q)\beta_{t,k}(q)\,\mathrm{d}q-\sum_{t=1}^{n-1}d_{t+1,k}\int_{0}^1(1-q)\beta_{t,k}(q)\, \mathrm{d}q\le d_{1,k},\label{ineq:weak-duality-4}
\end{align}
where the first inequality holds by constraint \eqref{const:constr_1_P_n_k-relax} and by changing the index range in the first summation, and the second inequality holds by inequality \eqref{const:constr_t_P_n_k-relax} and the fact that $d_{n+1,k}=0$.
On the other hand, note that 
\begin{align}
 &\sum_{t=1}^n\sum_{\ell=1}^{k-1}d_{t,\ell}\int_{0}^1\beta_{t,\ell}(q)\,\mathrm{d}q - \sum_{t=1}^n\sum_{\ell=1}^{k}d_{t+1,\ell-1}\int_{0}^1q\beta_{t,\ell}(q)\, \mathrm{d}q  -  \sum_{t=1}^n\sum_{\ell=1}^{k-1}d_{t+1,\ell}\int_{0}^1(1-q)\beta_{t,\ell}(q)\, \mathrm{d}q\notag\\
 &=\sum_{t=1}^n\sum_{\ell=1}^{k-1}d_{t,\ell}\int_{0}^1\beta_{t,\ell}(q)\,\mathrm{d}q - \sum_{t=1}^n\sum_{\ell=1}^{k-1}d_{t+1,\ell}\int_{0}^1q\beta_{t,\ell+1}(q)\, \mathrm{d}q  -  \sum_{t=1}^n\sum_{\ell=1}^{k-1}d_{t+1,\ell}\int_{0}^1(1-q)\beta_{t,\ell}(q)\, \mathrm{d}q\notag\\
&=\sum_{t=1}^n\sum_{\ell=1}^{k-1}d_{t,\ell}\int_{0}^1\beta_{t,\ell}(q)\,\mathrm{d}q - \sum_{t=1}^{n-1}\sum_{\ell=1}^{k-1}d_{t+1,\ell}\left(\int_{0}^1q\beta_{t,\ell+1}(q)\, \mathrm{d}q  +\int_{0}^1(1-q)\beta_{t,\ell}(q)\, \mathrm{d}q\right)\notag\\
&\le \sum_{t=1}^n\sum_{\ell=1}^{k-1}d_{t,\ell}\int_{0}^1\beta_{t,\ell}(q)\,\mathrm{d}q - \sum_{t=1}^{n-1}\sum_{\ell=1}^{k-1}d_{t+1,\ell}\int_{0}^1\beta_{t+1,\ell}(q)\, \mathrm{d}q\notag\\
&=\sum_{\ell=1}^{k-1}d_{1,\ell}\int_{0}^1\beta_{1,\ell}(q)\, \mathrm{d}q\le \sum_{\ell=1}^{k-1}d_{1,\ell}\mathbf{1}_{k}(\ell)= 0,\label{ineq:weak-duality-5}
\end{align}
where the first equality holds by changing the index range of the second summation and $d_{t+1,0}=0$ for every $t\in [n]$, the second equality holds by factoring the summations and $d_{n+1,\ell}=0$ for every $\ell\in [k-1]$, the first inequality holds by inequality \eqref{const:constr_t_P_n_2-relax}, and the last inequality by constraint \eqref{const:constr_1_P_n_k-relax}.

Then, by summing over $t\in [n]$ and $\ell\in [k]$ in inequality \eqref{ineq:weak-duality-3}, we get 
\begin{align}
&\sum_{t=1}^n\sum_{\ell=1}^k\int_{0}^1\beta_{t,\ell}(q)\int_0^q h(u) \, \mathrm{d}u\,\mathrm{d}q\notag\\ 
&\leq \sum_{t=1}^n\sum_{\ell=1}^kd_{t,\ell}\int_{0}^1\beta_{t,\ell}(q)\,\mathrm{d}q -\sum_{t=1}^n\sum_{\ell=1}^kd_{t+1,\ell-1}\int_{0}^1q\beta_{t,\ell}(q)\, \mathrm{d}q  - \sum_{t=1}^n\sum_{\ell=1}^kd_{t+1,\ell}\int_{0}^1(1-q)\beta_{t,\ell}(q)\, \mathrm{d}q,\notag\\
&\le d_{1,k},\label{ineq:weak-duality-6}
\end{align}
where the second inequality comes from \eqref{ineq:weak-duality-4} and \eqref{ineq:weak-duality-5} together.
Finally, \eqref{ineq:weak-duality-2} and \eqref{ineq:weak-duality-6} imply that $v\le d_{1,k}$, which concludes the proof of the lemma.
\end{proof}

\noindent{\bf Second step: A truncated LP with a useful structure.} Consider the following infinite-dimensional linear program:
For $\overline{n}\leq n$, we consider the following LP
\begin{align}
 \sup \quad  v \hspace{2cm} &\tag*{\normalfont{\mbox{[D]$_{n,k}(\bar{n})$}}}\label{form:P_nk_restricted_n}\\
	\normalfont\text{s.t.} \quad  \int_0^1 \alpha_{t,k}(q)\, \mathrm{d}q + \int_0^1 \sum_{\tau < t} q \alpha_{\tau, k}(q)\, \mathrm{d}q &\leq 1,  \hspace{.3cm} \text{for all }t\in [\bar{n}],\label{const:constr_1_P_n_k_restricted_n} \\
	\int_0^1 \alpha_{t,\ell}(q) \, \mathrm{d}q + \int_0^1 \sum_{\tau< t} q \alpha_{\tau, \ell} (q)\, \mathrm{d}q &\leq \int_0^1 \sum_{\tau < t} q\alpha_{\tau,\ell+1}(q) \, \mathrm{d}q,\hspace{.3cm}   \text{for all } t\in [\bar{n}], \ell \in [k-1] ,\label{const:constr_t_P_n_k_restricted_n}\\
	 v g_{n,k}(u)&\leq \sum_{t=1}^{\barnk}\sum_{\ell=1}^k \int_u^1 \alpha_{t,\ell}(q) \, \mathrm{d}q,  \text{ for }u \in [0,1]  , \label{const:P_n_approx_restricted_n} \\
	\alpha_{t,\ell}(q)&\ge 0\hspace{.3cm} \text{ for all }q\in [0,1],t\in [\bar{n}]\text{ and }\ell\in [k],\label{const:P_n_initial_restricted_n}
\end{align}

We will prove that the optimal value of~\ref{form:P_nk_restricted_n} is a lower bound to the optimal value of~\ref{form:dual-relaxation}. 
The following technical proposition allows us to ensure that any feasible solution to~\ref{form:P_nk_restricted_n} induces a feasible solution to~\ref{form:dual-relaxation}. We present the proof of the proposition in Appendix~\ref{app:LB}.
\begin{proposition}\label{prop:restricted_LP_tightens_constraints}
    For every feasible solution $(\alpha,v)$ to~\ref{form:P_nk_restricted_n}, there is $(\alpha',v)$ feasible to~\ref{form:P_nk_restricted_n} for which all constraints~\eqref{const:constr_1_P_n_k_restricted_n} and~\eqref{const:constr_t_P_n_k_restricted_n} are tightened.
\end{proposition}
The following proposition states the lower bound we need in the rest of our analysis.
\begin{proposition}
    For every $k<n$, { and every $\bar{n}\le n$}, the optimal value of~\ref{form:P_nk_restricted_n} is at most the optimal value of~\ref{form:dual-relaxation}.
\end{proposition}

\begin{proof}
    Let $(\alpha,v)$ be a feasible solution to~\ref{form:P_nk_restricted_n}. By the previous proposition, we can assume that $\alpha$ tightens all Constraints~\eqref{const:constr_1_P_n_k_restricted_n} and~\eqref{const:constr_t_P_n_k_restricted_n}. From here, we can deduce that for $0\leq t<\barnk$
    \begin{align*}
        \int_0^1 \alpha_{t+1,k}(q) \, \mathrm{d}q & = 1 - \sum_{\tau\leq  t}\int_0^1 q \alpha_{\tau, k}(q)\, \mathrm{d}q = \int_0^1 \alpha_{t,q}\, \mathrm{d}q - \int_0^1 q \alpha_{t,q}\, \mathrm{d}q
    \end{align*}
    and for $\ell < k$, we have
    \begin{align*}
        \int_0^1 \alpha_{t+1,\ell}(q)\, \mathrm{d}q & = \sum_{\tau\leq t}\int_0^1 q \alpha_{\tau, \ell+1}(q)\, \mathrm{d}q - \sum_{\tau\leq t}\int_0^1 q \alpha_{\tau, \ell}(q)\, \mathrm{d}q \\
        & = \int_0^1 \alpha_{t,q}\, \mathrm{d}q + \int_0^1 q\alpha_{t,\ell+1}(q)\, \mathrm{d}q - \int_0^1 q \alpha_{t,\ell}(q)\, \mathrm{d}q\\
        & = \int_0^1 (1-q) \alpha_{t,\ell}(q)\, \mathrm{d}q + \int_0^1 q \alpha_{t,\ell+1}(q)\, \mathrm{d}q.
    \end{align*}
    Hence, $\alpha$ satisfies constraints~\eqref{const:constr_1_P_n_k-relax}~\eqref{const:constr_t_P_n_k-relax}~\eqref{const:constr_t_P_n_2-relax} for $t<\barnk$. If we define $\bar{\alpha}$ as follows
    \[
    \bar{\alpha}_{t,\ell}(q) = \begin{cases}
        \alpha_{t,\ell}(q), & t\leq \barnk,\\
        0, & t> \barnk,
    \end{cases}
    \]
    then, $(\bar{\alpha},v)$ is a feasible solution to~\ref{form:dual-relaxation}. From here, the result follows immediately. 
\end{proof}

\noindent{\bf Third step: From the nonlinear system to LP.} 
Since $\gamma_{n,k}$ is equal to the value of $\ref{form:LP_dual}$, which in turn is at least the value of $\ref{form:dual-relaxation}$, the previous result implies that we only need to provide a feasible solution to~\ref{form:P_nk_restricted_n} to provide a lower bound on $\gamma_{n,k}$. The latter will be defined by a solution to the non-linear system of equations (\ref{ode1})-(\ref{ode3}). For a given $\theta$, we denote it by $\nls_k(\theta)$. 
The following lemma summarizes some properties of $\nls_k(\theta)$ that we use in our analysis.
\begin{lemma}\label{prop:nls}
For every positive integer $k$, the following holds:
\begin{enumerate}[itemsep=0pt,label=\normalfont(\roman*)]
    \item There exists $\theta^{\star}$ for which $\nls_k(\theta^{\star})$ has a solution.\label{prop:nls-theta}
    We denote such a solution by $(Y_1,\ldots,Y_k)$.
    \item The vector $\theta^{\star}$ satisfies that $0<\theta^{\star}_1<\theta^{\star}_2<\cdots<\theta^{\star}_k<1/k.$\label{prop:nls-diff-theta}
    \item For every $j\in [k]$, the function $Y_j$ is non-increasing.\label{prop:nls-y-dec}
\end{enumerate}
\end{lemma}
We defer the proof of Lemma~\ref{prop:nls} to Subsection \ref{sec:nls-lemma}.
Let $\barnk=n-k-1$ and let $y_{j,t}= Y_j(t/\barnk)$. 
{ Let us define $\varepsilon_{j,t} = - \ln(y_{j,t})/\barnk$, for $t\in \{0,1,\ldots,\barnk-1\}$, $j\in \{1,\ldots,k\}$, and $\varepsilon_{j,\barnk}=1$.} We can show that for $n$ large enough, $-\ln (y_{j,t})/\barnk \leq 1$ for $t\in \{0,\ldots,\barnk-1\}$ (see Proposition~\ref{prop:mon-phi}); hence, $0\leq \varepsilon_{1,j}\leq \cdots \leq \varepsilon_{j,\barnk}\leq 1$. Let $B_\ell = (\ell-1)\cdot (4c_k^k + c_k/k!)$ for $\ell\in \{1,\ldots,k\}$, where $c_k= 24k!\max\left\{ \theta_{\ell+1}^\star / \theta_\ell^\star : \ell\in \{1,\ldots,k-1 \right\}\}$. Now, consider the following family of functions:
\begin{align}
\alpha^\star_{t,\ell}(q) = \begin{cases}
    0, &  t\leq k-\ell,\\
    \left( 1 + 12\frac{\ln(\barnk)^2}{\barnk} \right)^{-(k-\ell+1)}\left( B_\ell \ln ({\barnk}) \mathbf{1}_{[0,1/\barnk]}(q) - \theta^{\star}_{\ell} g_{n,k}'(q)\mathbf{1}_{(0, \varepsilon_{\ell,t})}(q) \right), & t=k-\ell+1, \\
     \left(1 + 12 \frac{\ln(\barnk)^2}{\barnk} \right)^{-(k-\ell+1)} \left(-\theta^{\star}_{\ell} g_{n,k}'(q)\right)\mathbf{1}_{(\varepsilon_{\ell,t-1}, \varepsilon_{\ell,t})}(q), & t \geq k-\ell +2. \label{eq:description_of_solution}
\end{cases}
\end{align}
Note that for all $u\in [0,1]$, we have
\begin{align}
\sum_{\ell=1}^k \int_u^1 \alpha^\star_{t,\ell}(q)\,\mathrm{d}q &\geq \left( 1+ 12 \frac{\ln(\barnk)^2}{\barnk}  \right)^{-k}\left( \sum_{\ell=1}^k \theta^{\star}_{\ell} \right)g_{n,k}(u)\notag\\
&\geq \left(  1- 12 k\frac{\ln(\barnk)^2}{\barnk}\right) \left( \sum_{\ell=1}^k \theta_{\ell}^\star \right)g_{n,k}(u),\label{eq:theta_approx}
\end{align}
where in the first inequality we used that $(1+12\ln(\barnk)^2/\barnk)^{-(k-\ell+1)}$ is increasing in $\ell$ and in the second inequality we used the standard Bernoulli inequality. 
Inequality~\eqref{eq:theta_approx} guarantees that $(\alpha^\star,v^\star)$ satisfies constraint~\eqref{const:P_n_approx_restricted_n} with $v^\star=\left(  1- 12 k\cdot {\ln(\barnk)^2}/{\barnk}\right) \sum_{\ell=1}^k \theta_{\ell}^\star$. 
Before proving Theorem \ref{thm:comp-ratio} we need the following lemma; we defer the proof to section~\ref{sec:apx-lemma}.
\begin{lemma}\label{lem:main-apx-lemma}
    For $\barnk = n-k-1$, and $n$ large enough, $\alpha^\star$ satisfies constraints~\eqref{const:constr_1_P_n_k_restricted_n} and~\eqref{const:constr_t_P_n_k_restricted_n}.
\end{lemma}

\begin{proof}[Proof of Theorem~\ref{thm:comp-ratio}]
    As a consequence of Lemma~\ref{lem:main-apx-lemma}, we have that $(\alpha^\star,v^\star)$ is a feasible solution to~\ref{form:P_nk_restricted_n}. In particular, we obtain the approximation 
$$\gamma_{n,k} \geq v^\star = \left(  1- 12 k\frac{\ln(\barnk)^2}{\barnk}\right)\sum_{\ell=1}^k \theta_{\ell}^\star \geq \left( 1 - 24 k \frac{\ln(n)^2}{n}\right)\sum_{\ell=1}^k \theta_{\ell}^\star,$$
when $n$ is sufficiently large.
\end{proof}

\color{black}

\subsection{Analysis of $\nls_k(\theta)$ and Proof of Lemma \ref{prop:nls}}\label{sec:nls-lemma}
In this section, we analyze the nonlinear system $\nls_k(\theta)$ in terms of the existence of solutions. 
Given functions $y_1,\ldots,y_k:\RR\to \RR_+$, let $y=(y_1,\ldots,y_k)$, and for each pair $r,\ell\in [k]$, define $\phi_{r,\ell,y}(t)=\Gamma_r(-\ln y_{\ell}(t))$ for every $t\in [0,1)$.
Observe that by simple differentiation, we have 
\begin{equation}\phi_{r,\ell,y}'(t)=-\Gamma_r'(-\ln y_{\ell}(t))\frac{y_{\ell}'(t)}{y_{\ell}(t)}=(-\ln y_{\ell}(t))^{r-1}y_{\ell}'(t),\label{phi-derivative}\end{equation}
since $\Gamma_r'(x)=-x^{r-1}e^{-x}$. Furthermore, when $r\ge 2$, observe that $\phi_{r,\ell,y}'(t)=-\phi_{r-1,\ell,y}'(t)\ln y_{\ell}(t)$, which is a consequence of the derivative formula in \eqref{phi-derivative}.
For a vector $\theta_{\ell:k}=(\theta_\ell,\ldots,\theta_k)$, we define the system $\nls_{\ell,k}(\theta_{\ell:k})$ to be the subsystem of $\nls_k(\theta)$ that only consider the differential equations from $\ell,\ldots,k$ and the terminal conditions, that is, 
\begin{align}
    (\Gamma_k(-\ln y_k))' & =  k! \left( 1 - 1/(k\theta_k) \right) - \Gamma_{k+1}(-\ln y_k),\notag\\
    (\Gamma_{k}(-\ln y_j))' & =  k!-\Gamma_{k+1}(-\ln y_j) - \frac{\theta_{j+1}}{\theta_j}(k!-\Gamma_{k+1}(-\ln y_{j+1}))\text{ for every } j\in \{\ell,\ldots,k-1\},\notag\\
    y_j(0) & =1 \text{ and }\lim_{t\uparrow 1}y_j(t) =0  \hspace{.4cm}\text{for every } j\in \{\ell,\ldots,k\}.\notag
\end{align}
When $\ell=1$, the system $\nls_{1,k}(\theta)$ is exactly the system $\nls_k(\theta)$.
We also remark that, by replacing, any solution $y$ of $\nls_{\ell,k}(\theta_{\ell:k})$ satisfies the following conditions:
\begin{align}
    \phi_{k,k,y}' & =  k! \Big( 1- \frac{1}{k\theta_k}\Big) - \phi_{k+1,k,y},\notag\\
    \phi_{k,j,y}' & =  k!-\phi_{k+1,j,y} - \frac{\theta_{j+1}}{\theta_j}(k!-\phi_{k+1,j+1,y})\;\text{ for every } j\in \{\ell,\ldots,k-1\}.\label{ode-new-phi}
\end{align}
We will use $\nls_{\ell,k}(\theta_{\ell:k})$ to inductively show that $\nls_k{(\theta)}$ satisfies all properties of Lemma \ref{prop:nls}. One key step to showing the existence of a solution to $\nls_{\ell,k}(\theta_{\ell:k})$ is to first establish some properties that any solution of $\nls_{\ell+1,k}(\theta_{\ell+1:k})$, provided by the induction hypothesis, must satisfy. In Proposition \ref{prop-inductive-step} we show \ref{prop:nls}\ref{prop:nls-diff-theta} and \ref{prop:nls}\ref{prop:nls-y-dec}. Proposition \ref{prop:monotonicity} is useful for showing \ref{prop:nls}\ref{prop:nls-y-dec} as it gives a simple sufficient criterion for monotonicity to hold.
The following proposition gives an equivalent formulation of some expressions used to understand the behaviour of $\phi_{k,j,y}'$.
\begin{proposition}\label{prop:phi-identity}
Consider $\ell\in [k-1]$, and let $\theta_{\ell,k}$ be such that there is a solution $y=(y_{\ell},\ldots,y_k)$ for $\nls_{\ell,k}(\theta_{\ell:k})$, and such that $y_{j}'(s)\ne 0$ for every $j$ and every $s\in (0,1)$. 
Then, the following holds:
\begin{enumerate}[itemsep=0pt,label=(\roman*)]
    \item $\displaystyle\phi_{k,k,y}'(t)\exp\Big(\int_t^1 \ln y_k(s) \, \mathrm{d}s\Big) = k! \left( 1- \frac{1}{k\theta_k}\right).$\label{phi-identity-k}
    \item For every $j\in \{\ell,\ldots,k-1\}$, we have that $\displaystyle\phi_{k,j,y}'(t)\exp\Big(\int_t^1 \ln y_{j}(s)\, \mathrm{d}s\Big)$ is equal to 
    $$ k!\left( 1 - \frac{\theta_{j+1}}{\theta_j} \right)+ \frac{\theta_{j+1}}{\theta_j} \int_t^1 \phi_{k,j+1,y}'(\tau)\ln y_{j+1}(\tau) \exp\Big(\int_\tau^1 \ln y_{j}(s)\, \mathrm{d}s\Big) \, \mathrm{d}\tau.$$\label{phi-identity-j}
\end{enumerate}
\end{proposition}

\begin{proof} We start by observing the following: $\phi_{k,k,y}'' = -\phi_{k+1,k,y}' = -(-\ln y_k)^{k}y_k' =\phi_{k,k,y}'\ln y_k$, 
where the first equality holds from the first identity in \eqref{ode-new-phi}, and the other two equalities come from \eqref{phi-derivative}.
From here, by integrating, we have that for every $t,r\in (0,1)$ with $r\ge t$, it holds 
\begin{align}
\phi'_{k,k,y}(t)\exp\Big(\int_{t}^{r}\ln y_k(s)\, \mathrm{d}s\Big)&=\phi'_{k,k,y}(t)\exp\Big(\int_{t}^{r}\frac{\phi_{k,k,y}''(s)}{\phi_{k,k,y}'(s)}\, \mathrm{d}s\Big)\notag\\
&=\phi'_{k,k,y}(t)\exp\Big(\ln \phi'_{k,k,y}(r)-\ln \phi'_{k,k,y}(t)\Big)=\phi'_{k,k,y}(r).\notag
\end{align}
We conclude part \ref{phi-identity-k} by doing $r\to 1$: We use that $y_k(r)\to 0$ in $\nls_k(\theta)$, therefore $\phi_{k+1,k,y}(r)\to 0$, and then $\phi_{k,k,y}'(r)\to k!(1-1/(k\theta_k))$, using the first equality in \eqref{ode-new-phi}.

For $j\in \{\ell,\ldots,k-1\}$, we proceed in a similar way.
From the second equality in \eqref{ode-new-phi} we get
    \begin{align}
        \phi_{k,j,y}'' & = - \phi_{k+1,j,y}' + \frac{\theta_{j+1}}{\theta_j} \phi_{k+1,j+1,y}' = \phi_{k,j,y}'\ln y_k - \frac{\theta_{j+1}}{\theta_j} \phi_{k,j+1,y}' \ln y_{j+1},\label{phi''}
    \end{align}
    where the last equality comes from the observation after the derivative formula in \eqref{phi-derivative}.
    On the other hand, for every $r,\tau\in (0,1)$ with $r\ge \tau$, we have
    \begin{align}
    &\frac{\partial}{\partial \tau}\left(\phi_{k,j,y}'(\tau) \exp\Big({\int_\tau^r \ln y_{j}(s)\, \mathrm{d}s}\Big)\right)\notag \\
    &=\phi_{k,j,y}''(\tau)\exp\Big({\int_\tau^r \ln y_{j}(s)\, \mathrm{d}s}\Big)-\phi_{k,j,y}'(\tau)\exp\Big({\int_\tau^r \ln y_{j}(s)\, \mathrm{d}s}\Big)\ln y_j(\tau)\notag\\
    &=\Big(\phi_{k,j,y}''(\tau)-\phi_{k,j,y}'(\tau)\ln y_j(\tau)\Big)\exp\Big({\int_\tau^r \ln y_{j}(s)\, \mathrm{d}s}\Big)\notag\\
    &=-\frac{\theta_{j+1}}{\theta_j} \phi_{k,j+1,y}'(\tau)\ln y_{j+1}(\tau)\exp\Big({\int_\tau^r \ln y_j(s)\, \mathrm{d}s}\Big),\notag
    \end{align}
    where the last equality comes from the equality in \eqref{phi''}.
    We conclude part \ref{phi-identity-j} by doing $r\to 1$ and then integrating $\tau$ between $t$ and one: We use that $y_j(r)\to 0$ in $\nls_k(\theta)$, therefore $\phi_{k+1,j,y}(r)\to 0$, $\phi_{k+1,j+1,y}(r)\to 0$, and then $\phi_{k,j,y}'(r)\to k!(1-\theta_{j+1}/\theta_j)$, using the second equality in \eqref{ode-new-phi}.
\end{proof}
\begin{proposition}\label{prop:monotonicity}
Consider $\ell\in [k-1]$, and let $\theta_{\ell:k}$ be such that there is a solution $y=(y_{\ell},\ldots,y_k)$ for $\nls_{\ell,k}(\theta_{\ell:k})$, and let $j\in \{\ell,\ldots,k-1\}$.
If $y_{j+1}$ is non-increasing, and if there is $t_1\in [0,1)$ such that $y_j'(t_1)<0$ and $y_j(t_1) <1$, then $y_{j}'(t) < 0$ for all $t\in [t_1,1)$.
\end{proposition}

\begin{proof}
We prove the result by contradiction. 
Suppose there exists $t_2\in (t_1,1)$ such that $y_j'(t_2)\geq 0$. 
By the continuity of $y_j$, the value $\min\{ y_j(t): t\in [t_1,t_2]  \}$ is well-defined, the minimum in $[t_1,t_2]$ is attained at $t'\in [t_1,t_2]$, and $y_j(t')<1$ since $y_j(t_1)<1$. Then, in a neighborhood of $t'$, there is $t_1'<t_2'$ such that $y_j(t_1')=y_j(t_2') < 1$ and $y_j'(t_1')<0$ and $y_j'(t_2')\geq 0$. Then,
\begin{align*}
        0  > (-\ln y_{j}(t_1'))^{r-1}y_{j}'(t_1')=\phi_{k,j,y}'(t_1')  &= k! \left( 1 - \frac{\theta_{j+1}}{\theta_j} \right) - \phi_{k+1,j,y}'(t_1') + \frac{\theta_{j+1}}{\theta_j}\phi_{k+1,j+1,y}(t_1') \\
        & \geq k! \left( 1 - \frac{\theta_{j+1}}{\theta} \right) - \phi_{k+1,j,y}'(t_2') + \frac{\theta_{j+1}}{\theta_j}\phi_{k+1,j+1,y}(t_2')\\
        & = \phi_{k,j,y}'(t_2'),
\end{align*}
where in the second inequality we used that $y_{j+1}$ is non-increasing. From here, the contradiction follows since $\phi_{k,j,y}'(t_2')=(-\ln y_{j}(t))^{r-1}y_{j}'(t_2')\ge 0$.
\end{proof}

\begin{proposition}\label{prop-inductive-step}
Consider $\ell\in [k-1]$, and let $\theta_{\ell:k}$ be such that there is a solution $y=(y_{\ell},\ldots,y_k)$ for $\nls_{\ell,k}(\theta_{\ell:k})$.
Then, the following conditions are necessary: For every $j\in \{\ell,\ldots,k\}$, $y_j$ is strictly decreasing in $[0,1)$, $\theta_j<\theta_{j+1}$ for all $j<k$, and $\theta_k<1/k$.
\end{proposition}

\begin{proof}
We proceed by induction.
Since $y_k(0)=1$, $y_k(t)\to 0$ for $t\rightarrow 1$ from the left, and $y_k$ is differentiable in $(0,1)$, there is a value $t_k\in (0,1)$ such that $y_k'(t_k)<0$ and $y_k(t_k) < 1$. 
Since $\phi_{k,k,y}'(t)=y_k'(t)(-\ln y_k(t))^{k-1}$, we have $\phi_{k,k,y}'(t_k) < 0$. 
Then, from Proposition \ref{prop:phi-identity}\ref{phi-identity-k}, it must be that $1-1/k\theta_k<0$, that is, $\theta_k < 1/k$. 
Together with Proposition \ref{prop:phi-identity}\ref{phi-identity-k}, this implies that for every $t\in (0,1)$ it holds $\phi_{k,k,y}'(t) < 0$, that is, $\phi_{k,k,y}=\Gamma_k(-\ln y_k)$ is strictly decreasing in $(0,1)$. 
We conclude that $y_k$ strictly decreases in $[0,1)$. 

Assume inductively that $y_{j+1},\ldots,y_k$ are strictly decreasing for some $j<k$. 
We will show that $\theta_j< \theta_{j+1}$ and $y_j$ is strictly decreasing. 
Note that for every $t\in (0,1)$, we have
\begin{equation}
\phi_{k,j+1,y}'(t) \ln y_{j+1}(t) = -y_{j+1}'(t)(-\ln y_{j+1}(t))^k > 0,\label{induction-ineq}
\end{equation}
where the inequality follows by our inductive assumption and the equality by the derivative formula in \eqref{phi-derivative}. Now, if $\theta_{j+1}\le  \theta_{j}$, for every $t\in (0,1)$ we have
\begin{align}
&\phi_{k,j,y}'(t)\exp\Big(\int_t^1 \ln y_{j}(s)\, \mathrm{d}s\Big)\notag\\
&=k!\left( 1 - \frac{\theta_{j+1}}{\theta_j} \right)+ \frac{\theta_{j+1}}{\theta_j} \int_t^1 \phi_{k,j+1,y}'(\tau)\ln y_{j+1}(\tau) \exp\Big(\int_\tau^1 \ln y_{j}(s)\, \mathrm{d}s\Big) \, \mathrm{d}\tau\notag\\
&\ge \frac{\theta_{j+1}}{\theta_j} \int_t^1 \phi_{k,j+1,y}'(\tau)\ln y_{j+1}(\tau) \exp\Big(\int_\tau^1 \ln y_{j}(s)\, \mathrm{d}s\Big) \, \mathrm{d}\tau\ge 0,\notag
\end{align}
where the first equality holds by Proposition \ref{prop:phi-identity}\ref{phi-identity-j}, the first inequality holds by $\theta_{j+1}\le \theta_j$, and in the last inequality we used inequality \eqref{induction-ineq} and the inductive assumption. 
Therefore, for every $t\in (0,1)$, we have $\phi_{k,j,y}'(t)\geq 0$, which cannot happen since the differentiability of $y_j$ and the border conditions imply that we can always find $t_j\in (0,1)$ such that $y_j'(t_j)<0$ and $y_j(t_j)<1$, i.e., 
$$\phi_{k,j,y}'(t_j)=y_j'(t_j)(-\ln y_j(t_j))^{k-1}<0.$$
We conclude that $\theta_j < \theta_{j+1}$. 

We prove next the monotonicity of $y_j$.
Consider $t'=\inf\{ t\in [0,1] : y_j'(t)<0  , y_j(t)<1  \}$, which is well-defined since the set is non-empty. 
If $t'=0$, then there is a sequence $(t_n)_{n\in \NN}$ in $(0,1)$ such that $t_n'\to 0$, $y_j'(t_n)<0$, and $y_j(t_n)<1$ for all $n\in \NN$. 
Then, since $y_{j+1}$ is strictly decreasing, by Proposition \ref{prop:monotonicity} we get that for all $n\in \NN$ and every $t\in [t_n',1)$ we have $y_j'(t)<0$. 
Since $t_n'\to 0$, we conclude that $y_j'$ is strictly decreasing in $[0,1)$.
Otherwise, suppose that $t'>0$. 
Then, $y_j'(t)\geq 0$ or $y_j(t)\geq 1$ for every $t\in (0,t')$. 
Assume that $y_j(s)>1$ for some $s\in (0,t')$. 
Then, since $\lim_{q\to 1}y_j(q)=0$, the continuity of $y_j$ and the fact that $y_j(t')\le 1$, imply the existence of a value $t''\in (0,t']$ such that $y_j(t'')=1$. 
Note that $0=y_j'(t'') (-\ln y_j(t'') )^{k-1}=\phi_{k,j,y}'(t'')$, and
\begin{align*}
\phi_{k,j,y}'(t'') &=k!-\phi_{k+1,j,y}(t'') -\frac{\theta_{j+1}}{\theta_j}\left(k!- \phi_{k+1,j+1}(t'') \right)\\
&=k!-\Gamma_{k+1}(0) -\frac{\theta_{j+1}}{\theta_j}\left(k!- \phi_{k+1,j+1}(t'') \right)\\
&=-\frac{\theta_{j+1}}{\theta_j}\left(k!- \Gamma_{k+1}(-\ln y_{j+1}(t'')) \right)< 0,
\end{align*}
which is a contradiction; the first equality holds from \eqref{ode-new-phi}, the second holds since $y_j(t'')=1$, the third since $\Gamma_{k+1}(0)=k!$, and the inequality follows from $y_{j+1}$ being strictly decreasing.
Therefore, we have $y_j(t)\leq 1$ for all $t\in (0,t')$, which further implies that $y_j(t)\leq 1$ for all $t\in (0,1)$. 

If $y_j(s)<1$ for some $s\in [0,t']$, then there exists $t'''\in (0,t')$ such that $y_j'(t''')<0$ and $y_j(t''')<1$, which contradicts the minimality of $t'$. 
Then, $y_j(t)=1$ for every $t\in (0,t']$.
But this implies that $\phi'_{k,j,y}(t)=0$ and $\phi_{k+1,j,y}(t)=k!$ for every $t\in (0,t']$, and therefore from \eqref{ode-new-phi} we get that $k!=\phi_{k+1,j+1,y}(t)=\Gamma_{k+1}(-\ln y_{j+1}(t))$ for every $t\in (0,t']$.
This implies that $y_{j+1}(t)=1$ for every $t\in (0,t']$, which contradicts the fact that $y_{j+1}$ is strictly decreasing.
We conclude that $t'=0$ and, therefore, $y_j$ is strictly decreasing.
This finishes the proof of the proposition.
\end{proof}

\begin{proof}[Proof of Lemma \ref{prop:nls}]
In what follows, we show that there is a choice of $\theta$ such that the system $\nls_k(\theta)$ has a solution.
We proceed inductively. 
We show that there is a solution to this system for an appropriate choice of $\theta_{j:k}$. Using this solution, we can extend it to a solution for $\nls_{j-1,k}(\theta_{j-1:k})$, where $\theta_{j-1:k}=(\theta_{j-1},\theta_{j:k})$ for an appropriate choice of $\theta_{j-1}$.
For every $\ell$ we denote by $y_{\ell}(1)$ the value $\lim_{t\uparrow 1}y_{\ell}(t)$.

We start with $j=k$. In this case, the $\nls_{k,k}(\theta_k)$ is the following system:
\begin{align}
    \Gamma_{k}(-\ln y_k)' & = k!-k!/(\theta_k k) - \Gamma_{k+1}(-\ln y_k),\notag\\
y_k(0) & = 1\text{ and } y_k(1)=0.\notag
\end{align}
We can analyze this system in the same way as the Hill and Kertz differential equation when $k=1$ (see, e.g., \cite{correa2021posted,brustle2022competition}).
There is a solution to this system if and only if $\theta_k$ satisfies the following integral equation:
\[
1 = \int_0^1 \frac{(-\ln y)^{k-1}}{ (k!/(k\theta_k)) - k! + \Gamma_{k+1}(-\ln y)  }\, \mathrm{d}y.
\]
This holds by noting that $\Gamma_{k}(-\ln y_k)'=(-\ln y_{k})^{k-1}y_{k}'$; we integrate over $[0,1]$ and use the border conditions in the system, and perform a change of variables.
From here, we get a unique $\theta^{\star}_k$ that satisfies the requirements since the value of the integral is monotone as a function of $\theta_k$, and there exists one for which the integral is exactly equal to 1.
Furthermore, there exists a unique solution $Y_k$ to $\nls_{k,k}(\theta_k^{\star})$.

Assume inductively that we have found a $\theta^{\star}_{j+1:k}$ where we have a solution $(Y_{j+1},\ldots,Y_k)$ to $\nls_{j+1,k}(\theta^{\star}_{j+1:k})$ for some $j<k$.
We now show that the system $\nls_{j,k}(\theta_j,\theta^{\star}_{j+1},\ldots,\theta^{\star}_k)$ is feasible for a choice of $\theta_j$. 
This boils down to finding a solution for the following: 
\begin{align}
    \Gamma_{k}(-\ln Y_j)' & = k! - \Gamma_{k+1}(-\ln Y_j) - \frac{\theta^{\star}_{j+1}}{\theta_j}(k! -\Gamma_{k+1}(-\ln Y_{j+1})) \label{eq:system_for_j} \\
    Y_j(0) & = 1, \text{ and }Y_{j}(1)=0. \label{eq:system_for_j_initial_cond}
\end{align}
where $\theta^{\star}_{j+1}$ and $Y_{j+1}$ are given and satisfy $Y_{j+1}(0)=1, Y_{j+1}(1)=0$. 
By Proposition \ref{prop-inductive-step} we have that $Y_{j+1},\ldots,Y_k\in [0,1]$ are strictly decreasing and $\theta^{\star}_{j+1}< \cdots < \theta^{\star}
_{k }< 1/k$.

Let $\theta_{j}>0$ and consider the following Euler approximation to a candidate solution to \eqref{eq:system_for_j}-\eqref{eq:system_for_j_initial_cond}. Let $m\in \ZZ_+$ be non-negative and consider the following recursion: $y_{m,j,0} = 1$, and $\Gamma_{k}(-\ln y_{m,j,t+1})$ is equal to
\begin{equation}\Gamma_{k}(-\ln y_{m,j,t}) + \frac{1}{m}\left( k! - \Gamma_{k+1}(-\ln y_{m,j,t}) - \frac{\theta^{\star}_{j+1}}{\theta_{j}}\Big(k! - \Gamma_{k+1}(-\ln Y_{j+1}(t/m))\Big)   \right).\label{eq:Euler_iterate}\end{equation}
Note that the sequence
is well-defined for $y_{m,j,t}\geq 0$. Let $t'=\max\{ t\in [m]\cup\{0\} : y_{m,j,t} \geq 0 \}$. 
For $t=0$, we have $\Gamma_{k}(-\ln y_{m,j,1}) = (k-1)!$ and therefore $y_{m,j,1}=1$.
For $t=1$, we have
\begin{align}\label{eq:euler-approx_t=1}
\Gamma_{k}(-\ln {y_{m,j,2}}) = (k-1)!  - \frac{\theta^{\star}_{j+1}}{\theta_j m }(k! - \Gamma_{k+1}(-\ln Y_{j+1}(1/m))) < (k-1)!,
\end{align}
which implies that $y_{m,j,2}< 1$ for any $\theta_j>0$. 
We note that if $\theta_j \to \infty$, then $\Gamma_{k}(-\ln y_{m,j,2}) \to (k-1)!$. Inductively, we can show that for $\theta_{j}\to \infty$, ${y_{m,j,t}=1}$ for all $t$; in particular, $t'=m$.

We now show that $y_{m,j,t}$ is decreasing in $t$ as long as $y_{m,j,t}\geq 1/m$ and $m$ is such that $m/\ln (m)\geq 1$ which holds for $m\geq 2$. We know this is true for $t\in \{1,2\}$. We assume the result holds from 1 up to $t$, and we show next the result holds for $t+1$, with $t\geq 2$.
Observe that
\begin{align}\label{eq:euler-approx_difference}
    &\Gamma_{k}(-\ln y_{m,j,t+1}) - \Gamma_{k}(-\ln y_{m,j,t}) \\
    & = \frac{1}{m} \left( k! - \Gamma_{k+1}(-\ln y_{m,j,t}) - \frac{\theta^{\star}_{j+1}}{\theta_j}(k!- \Gamma_{k+1}(-\ln Y_{j+1} (t/m) )) \right) \label{eq:euler_approx_second_line} \\
    & = \frac{1}{m} \sum_{\tau=0}^{t-1}(\Gamma_{k+1}(-\ln y_{m,j,\tau}) -\Gamma_{k+1}(-\ln y_{m,j,\tau+1}) ) \nonumber\\
    & \quad - \frac{\theta^{\star}_{j+1}}{m\theta_j}\sum_{\tau=0}^{t-1}(\Gamma_{k+1}(-\ln Y_{j+1}(\tau/m)) -\Gamma_{k+1}(-\ln Y_{j+1}((\tau+1)/m)) ) \nonumber \\
    & = \Gamma_{k}(-\ln y_{m,j,t}) - \Gamma_{k}(-\ln y_{m,j,t-1}) + \frac{1}{m}(\Gamma_{k+1} (-\ln y_{m,j,t-1}) -\Gamma_{k+1}(-\ln y_{m,j,t}) ) \label{eq:euler_approx_second_last_line}\\
    & \quad - \frac{\theta^{\star}_{j+1}}{m\theta_j}( \Gamma_{k+1}( -\ln Y_{j+1}((t-1)/m)) -\Gamma_{k+1}(-\ln Y_{j+1}(t/m)))\nonumber ,
\end{align}
where the first equality holds by writing the previous expression using two telescopic sums, and the third equality holds by rearranging terms and using the Euler approximation recursion.
Note that $y_{m,j,t+1} < y_{m,j,t}$ if and only if $\Gamma_{k}(-\ln y_{m,j,t+1}) < \Gamma_{k}(-\ln y_{m,j,t})$. Since $Y_{j+1}$ is strictly decreasing, the result follows after the following claim.
The proof of Claim~\ref{claim-euler} is in Appendix~\ref{app:LB}.

\begin{claim}\label{claim-euler}
    $\Gamma_{k}(-\ln y_{m,j,t}) - \Gamma_{k}(-\ln y_{m,j,t-1})
     + \frac{1}{m}(\Gamma_{k+1} (-\ln y_{m,j,t-1}) -\Gamma_{k+1}(-\ln y_{m,j,t}) )\leq 0$.
\end{claim}

We now show that $\partial y_{m,j,t}/\partial \theta_j \geq 0$ for all $t\leq t'$ and such that $y_{m,j,t}\geq 1/m$. Furthermore, we show that for $t\geq 1$ as before, we have $\partial y_{m,j,t}/\partial \theta_j > 0$. We proceed by induction in $t$. The result is clearly true for $t=0$. Suppose that $\partial y_{m,j,t}/\partial \theta_{j} \geq 0$ and let's show the result for $t+1$. By deriving in $\theta_j$ in the Euler recursion~\eqref{eq:Euler_iterate}, we have
\begin{align*}
    &(-\ln y_{m,j,t+1})^{k-1} \frac{\partial y_{m,j,t+1}}{\partial \theta_j}\\
    & = \frac{\partial\,\,}{\partial \theta_j} \Gamma_{k}(-\ln y_{m,j,t}) -\frac{1}{m}\frac{\partial \,\,}{\partial \theta_j} \Gamma_{k+1}(-\ln y_{m,j,t}) + \frac{\theta_{j+1}^\star}{\theta_j^2m}(k!-\Gamma_{k+1}( -\ln Y_{j+1}(t/m) )) \\
    & = \frac{\partial}{\partial \theta_j}\Gamma_{k}(-\ln y_{m,j,t})\left(1+ \frac{1}{m}\ln y_{m,j,t}\right) + \frac{\theta_{j+1}^\star}{\theta_j^2m}(k! - \Gamma_{k+1}(-\ln Y_{j+1}(t/m))) \\
    & \geq \frac{\theta_{j+1}^\star}{\theta_j^2m}(k!-\Gamma_{k+1}(-\ln Y_{j+1}(1/m)))
\end{align*}
where we used that {$y_{m,j,t}\geq 1/m$}. 
{Since $y_{m,j,t}\in (0,1)$, it follows that $\partial y_{m,j,t+1}/\partial \theta_j > 0$.} 
Notice that the right-hand side of the inequality is independent of $t$ and grows as $1/\theta_{j}^2$. 
Hence, as $\theta_j\to 0$, we have that $t'\to 1$ and so $y_{m,j,t}\to 1$ for $t\leq t'$. As a byproduct of this analysis, we also see that $t'$ is strictly increasing in $\theta_j$. 
Now, let $\theta_j(m)$ be such that $2/m\geq y_{m,j,m-\sqrt{m}}\geq 1/m$. The next claim shows that $\theta_j(m)\leq \theta_{j+1}^\star$. We defer its proof to Appendix~\ref{app:LB}.
\begin{claim}\label{claim:bounded_theta_finite_m}
We have $\theta_j(m)\leq \theta_{j+1}^\star$.
\end{claim}


From the claim, we have that $\{\theta_j(m)\}_m$ is bounded. Thus, if we let $m$ tend to infinity, we can find a convergent subsequence $\{\theta_j(m_{\ell})\}_{\ell}$ with a limit denoted as $\theta^{\star}_j$.

Let $y_{\ell,j}:[0,1]\to [0,1]$ be the piece-wise linear interpolation of the points $\{ y_{m_\ell,j,t} \}_{t}$, where $y_{m_{\ell},j,t}$ is assigned as the image to the point $t/m_{\ell}\leq 1$. By a standard argument we can show that the sequence $\{y_{\ell,j}\}_{\ell}$ has a uniformly convergent subsequence to a function $Y_j:[0,1]\to [0,1]$ (see, e.g., \citep[Chapter 3]{kolmogorov1975introductory}). Furthermore, this function $Y_j$ is differentiable and satisfies~\eqref{eq:system_for_j}-\eqref{eq:system_for_j_initial_cond}. Hence, we have found $\theta^{\star}_{j:k}$ such that the system $\nls_{k,j}(\theta^{\star}_{j:k})$ is feasible.
\end{proof}

\subsection{Feasibility Analysis and Proof of Lemma \ref{lem:main-apx-lemma}}\label{sec:apx-lemma}

In this subsection, we prove Lemma~\ref{lem:main-apx-lemma}.
The crux of the proof follows by analyzing the functions $\alpha_{t,j}(q) = (1+12 \ln (\barnk)^2/\barnk)^{k-j+1} \alpha_{t,j}^\star(q)$. These functions hold the following two claims:

\begin{claim}\label{claim:bound_for_alpha_k} There is $n_0\geq 1$ such that for any $n\geq n_0$ and for any $t\in \{0,\ldots,\barnk-1\}$, we have
\begin{align*}
    \int_0^1 \alpha_{t+1,k}(q) \, \mathrm{d}q + \sum_{\tau \leq t}\int_0^1 q \alpha_{\tau,k}(q)\, \mathrm{d}q \leq 1 + 12\frac{\ln(\barnk)^2}{\barnk}.
\end{align*}  
\end{claim}

\begin{claim}\label{claim:bound_for_alpha_t}
There is $n_0\geq 1$ such that for any $n\geq n_0$, for any $j< k$, and for any $t\in \{0,\ldots,\barnk-1\}$, we have
$$\int_0^1 \alpha_{t+1,j}(q)\, \mathrm{d}q+ \sum_{\tau\leq t}\int_0^1 q \alpha_{\tau,j}(q)\, \mathrm{d}q \leq \left( 1+ 12 \frac{\ln(\barnk)^2}{\barnk}\right)\sum_{\tau\leq t}\int_0^1 q\alpha_{\tau,j+1}(q)\, \mathrm{d}q.$$
\end{claim}
Using these two claims, we show how to conclude Lemma~\ref{lem:main-apx-lemma} and then prove them.
\begin{proof}[Proof of Lemma \ref{lem:main-apx-lemma}] 
First, we have,
\begin{align*}
    & \int_0^1 \alpha_{t+1,k}^\star (q) \, \mathrm{d}q + \sum_{\tau\leq t}\int_0^1 q \alpha_{\tau,k}^\star(q)\, \mathrm{d}q \\
    & = \frac{1}{1+12 \ln(\barnk)^2/\barnk}\left( \int_0^1 \alpha_{t+1,k}(q) \, \mathrm{d}q + \sum_{\tau \leq t}\int_0^1 q \alpha_{\tau,k}(q)\, \mathrm{d}q \right)\leq 1,
\end{align*}
where we used Claim~\ref{claim:bound_for_alpha_k}, which shows that $\alpha^\star$ satisfies constraints~\eqref{const:constr_1_P_n_k_restricted_n}. 
Additionally,
\begin{align*}
    &\int_0^1 \alpha_{t,j}^\star(q)\, \mathrm{d}q+ \sum_{\tau\leq t}\int_0^1 q \alpha_{\tau,j+1}^\star(q)\, \mathrm{d}q \\
    & = \frac{1}{(1+12 \ln(\barnk)^2/\barnk)^{k-j+1}} \left( \int_0^1 \alpha_{t,j}(q)\, \mathrm{d}q+ \sum_{\tau\leq t}\int_0^1 q \alpha_{\tau,j+1}(q)\, \mathrm{d}q  \right)  \\
    & \leq \frac{1}{(1+12 \ln(\barnk)^2/\barnk)^{k-j}} \sum_{\tau\leq t} \int_0^1 q \alpha_{\tau,j+1}(q)\, \mathrm{d}q \\
    & = \sum_{\tau\leq t}\int_0^1 q \alpha_{\tau,j+1}^\star (q)\, \mathrm{d}q,
\end{align*}
where in the inequality we used Claim~\ref{claim:bound_for_alpha_t}; which shows that $\alpha^\star$ satisfies constraints~\eqref{const:constr_t_P_n_k_restricted_n}. This concludes the lemma. 
\end{proof}
We devote the rest of this section to prove Claims~\ref{claim:bound_for_alpha_k} and~\ref{claim:bound_for_alpha_t}.
The two claims follow by a careful analysis of the solution to $\nls_{n,k}(\theta^\star)$ as well as the function $g_{n,k}$. In the following proposition, we leave some useful properties satisfied by the function $g_{n,k}$. 
Recall that we set $\barnk= n-k-1$. 
The proof can be found in Appendix \ref{app:LB}.

\begin{proposition}\label{prop:g-n-properties}
    For every $u\in (0,1)$, the following holds: 
    \begin{enumerate}[itemsep=0pt,label=\normalfont(\roman*)]
        \item $g_{n,k}'(u) = - (n-k+1)(n-k)\binom{n}{k-1} (1-u)^{n-k-1}u^{k-1}$. \label{gnk:diff-neg}
        \item $g_{n+1,k+1}'(u) =\frac{n+1}{k} u g_{n,k}'(u)$.\label{gnk:recursive}  
        \item If $n> (k+1)+2(k+1)^2$, $-g_{n,k}'(u) \leq - n \left( 1 + 4 \frac{k^2}{n}\right) \frac{\Gamma_{k}(\barnk u)'}{(k-1)!}.$\label{prop:limit_props_of_gnk}
        \item If $n>4k$ and $u\in (0,s)$, with $s\le \frac{1}{2\sqrt{\barnk}}$, then
    $-n \left( 1 - 4 \frac{k^2}{n}\right)\left( 1- \frac{\barnk s^2}{1-s} \right) \frac{\Gamma_{k}(\barnk u)'}{(k-1)!}\leq -g_{n,k}'(u).
    $\label{prop:limit_props_of_gnk-2}
    \end{enumerate}
\end{proposition}

Recall that for every $r,\ell\in [k]$ we defined $\Phi_{r,\ell}=\Gamma_{r}(-\ln(Y_{\ell}))$. 
Observe that conditions \eqref{ode1}-\eqref{ode2} for the nonlinear system $\nls_k(\theta^{\star})$ can be rewritten to get the following identities in $[0,1)$:
\begin{align}
    \Phi_{k,k}' & =  k! \left( 1 - 1/(k\theta^{\star}_k) \right) - \Phi_{k+1,k},\label{ode1-phi}\\
    \Phi_{k,\ell}' & =  k!-\Phi_{k+1,\ell} - \frac{\theta^{\star}_{\ell+1}}{\theta^{\star}_{\ell}}(k!-\Phi_{k+1,\ell+1})\text{ for every } \ell\in [k-1],\label{ode2-phi}.
\end{align}
Furthermore, since the functions $Y_j$ are non-increasing, $\Phi_{k,j}$ are also non-increasing.
Recall that $\Gamma_r'(x)=-x^{r-1}e^{-x}$ and therefore, when $r\ge 2$, we have $\Phi_{r,\ell}'(t)=-\Phi_{r-1,\ell}'(t)\ln Y_{\ell}(t)$.
We use the following technical proposition in the rest of our analysis. For the sake of presentation, we defer its proof to Appendix~\ref{app:LB}.
\begin{proposition}\label{prop:mon-phi}
For every positive integer $k$, the following holds:
\begin{enumerate}[itemsep=0pt,label=\normalfont(\roman*)]
    \item Let $b_k=4k!\max \left\{\theta_{\ell+1}^\star/\theta_\ell^\star: { \ell\in \{1,\ldots,k-1\}} \right\}$. Then, for every $t\in (0,1)$, every $\ell\in [k]$, and every $r \in \{0,\ldots,k-1\}$, we have $Y_{\ell}(t)(-\ln Y_{\ell}(t))^{r}\leq b_k (1-t)$.\label{prop:nls-b} 
    \item Let $d_k=\min\left\{ \theta_{\ell+1}^*/\theta_\ell^*: { \ell\in \{1,\ldots,k-1\}}\right\} -1 >0$. There exists $\Delta_k>0$ such that for every $t\in (\Delta_k,1]$ and $\ell \in [k]$, it holds $Y_j(t)\geq d_k (1-t)^2$.\label{prop:nls-delta}
    \item Let $c_k=6 b_k$. We have $\Phi_{k,k}''(t)\ge 0$ for every $t\in (0,1)$. Furthermore, there is an integer $N_k$, such that for every $n\geq N_k$, every $\ell\in [k-1]$, and every $t\in ( 0,1-1/n]$, we have $|\Phi_{k,\ell}''(t)|\leq c_k \ln (n)$.\label{prop:nls-nk}

    \item {  Let $\bar{c}_k = (k c_k)^{1/k}$ and  $N_k$ as in~\ref{prop:nls-nk}. There is $\delta_k>0$ such that for any $n\geq N_k$, $j < k$, and every $t\leq \min\{ \delta_k, 1-1/n\} $, we have $Y_j(t) \geq 1 - \bar{c}_k\ln(n)^{1/k} t^{2/k}$.}\label{prop:nls-lower_bound}
\end{enumerate}    
\end{proposition}

\begin{proof}[Proof of Claim~\ref{claim:bound_for_alpha_k}]
For $t\in \{0,1,\ldots,\barnk-1\}$, we have
\begin{align}
    &\int_{0}^{1} \alpha_{t+1,k}(q) \, \mathrm{d}q + \sum_{\tau \leq t} \int_0^1 q \alpha_{\tau,k} (q)\, \mathrm{d}q\notag\\
    &= \theta_k^\star \left( \int_{\varepsilon_{k,t}}^{\varepsilon_{k,t+1}} (-g_{n,k}'(q)) \, \mathrm{d}q + \int_0^{\varepsilon_{k,t}} q (-g_{n,k}'(q))\, \mathrm{d}q \right) + B_k \frac{\ln (\barnk)}{\barnk}. \label{eq:for_alpha_k}
\end{align}
Now, we bound the term in parenthesis:
\begin{align}
&  \int_{\varepsilon_{k,t}}^{\varepsilon_{k,t+1}} (-g_{n,k}'(q)) \, \mathrm{d}q + \int_0^{\varepsilon_{k,t}} q (-g_{n,k}'(q))\, \mathrm{d}q \nonumber \\
&= \int_{\varepsilon_{k,t}}^{\varepsilon_{k,t+1}} (-g_{n,k}'(u)) \, \mathrm{d}u + \frac{k}{n+1}\int_0^{\varepsilon_{k,t}} (-g_{n+1,k+1}'(u)) \, \mathrm{d}u \nonumber \\
&\leq  \left( 1 + 16 \frac{k^2}{n} \right) \left(  \frac{n}{(k-1)!} \int_{\varepsilon_{k,t}}^{\varepsilon_{k,t+1}} -(\Gamma_{k}(\barnk u))' \, \mathrm{d}u + \frac{k}{n+1} \frac{n+1}{k!} \int_0^{\varepsilon_{k,t}} - (\Gamma_{k+1}(\barnk u))' \, \mathrm{d}u  \right)\nonumber \\
&=  \frac{n}{(k-1)!} \left( 1 + 16 \frac{k^2}{n} \right) \left( \Gamma_k(-\ln y_{k,t}) - \Gamma_{k}(-\ln y_{k,t+1}) + \frac{1}{n}(k! - \Gamma_{k+1}(-\ln y_{k,t})   \right) \nonumber \\
&\leq \frac{n}{(k-1)!} \left( 1 + 16 \frac{k^2}{n} \right) \left( \frac{1}{\barnk}(k! - \frac{(k-1)!}{\theta^{\star}_k}  - \Gamma_{k+1}(-\ln y_{k,t})  )  - (\Gamma_{k}(-\ln y_{k,t+1})- \Gamma_{k}(-\ln y_{k,t})) \right) \nonumber \\
& \qquad + \frac{n}{\barnk} \left( 1+ 16 \frac{k^2}{n}\right) \frac{1}{\theta^{\star}_k} \nonumber \\
&= \frac{n}{(k-1)!} \left( 1 + 16 \frac{k^2}{n} \right) \Big(\Phi'_{k,k}(t/\barnk) - \barnk\Big(\Phi_{k,k}((t+1)/\barnk)) - \Phi_{k,k}(t/\barnk)\Big)\Big) \nonumber \\
& \qquad + \frac{n}{\barnk} \left( 1+ 16 \frac{k^2}{n}\right)\frac{1}{\theta^{\star}_k} \nonumber\\
&\leq  \frac{n}{\barnk(k-1)!} \left( 1+ 16 \frac{k^2}{n} \right) \left( \Phi_{k,k}'\left( \frac{t}{\barnk} \right) - \frac{\Phi_{k,k}((t+1)/\barnk)) - \Phi_{k,k}(t/\barnk)}{1/\barnk}  \right)\nonumber\\ 
&\qquad + \frac{1}{\theta^{\star}_{k}} \left( 1 + 20 \frac{k^2}{n}\right). \label{ineq:last_inequality_k}
\end{align}
The first equality and inequality follow by Proposition~\ref{prop:g-n-properties}, where we used implicitly that $1+4(k+1)^2/(n+1) \leq 1+16k^2/n$ for any $k\geq 1$. The next equality follows by computing the integrals. The next inequality follows by bounding $1/n\leq 1/\barnk$ and adding and subtracting $1/\theta^{\star}_k$. The last equality follows by rearranging terms and the last inequality follows by bounding $n/\barnk(1+16k^2/n)\le 1+20k^2/n$ for $n\geq 20k^2(k+1)/(4k^2-k-1)$ and any $k\geq 1$.

The following claim allows us to bound the first term in~\eqref{ineq:last_inequality_k}. We defer the proof of the claim to Appendix \ref{app:LB}.
\begin{claim}\label{claim:gamma_kk}
    It holds that
    $\displaystyle\Phi_{k,k}'\left( \frac{t}{\barnk} \right) - \frac{\Phi_{k,k}((t+1)/\barnk)) - \Phi_{k,k}(t/\barnk)) }{1/\barnk} \leq 0.$
\end{claim}

Then, in~\eqref{eq:for_alpha_k}, we have
\begin{align*}
    &\int_0^1 \alpha_{t+1,k}(q) \, \mathrm{d}q + \sum_{\tau \leq t}\int_0^1 q \alpha_{\tau,k}(q)\, \mathrm{d}q
    \leq  \left(1+20 \frac{k^2}{n}\right) + B_k \frac{\ln(\bar{n})}{\bar{n}} \leq 1 + 12 \frac{\ln (\barnk)^2}{\barnk},
\end{align*}
where the last inequality holds for $n$ large enough. This concludes the proof of Claim~\ref{claim:bound_for_alpha_k}.
\end{proof}

\begin{proof}[Proof of Claim~\ref{claim:bound_for_alpha_t}]
    For $j < k$, and $t= k-j$, we have
\begin{align*}
    &\int_0^1 \alpha_{t+1,j}(q)\, \mathrm{d}q + \sum_{\tau \leq  t} \int_0^1 q \alpha_{\tau,j}(q)\, \mathrm{d}q \leq B_j \frac{\ln(\barnk)}{\barnk} + \theta_j^\star \int_0^{\varepsilon_{j,k-j+1}} (-g_{n,k})'(q)\, \mathrm{d}q \\
    & = B_j \frac{\ln (\barnk)}{\barnk} + \theta_j^\star \int_{\varepsilon_{j,k-j}}^{\varepsilon_{j,k-j+1}}(- g_{n,k})'(q)\, \mathrm{d}q + \theta_j^\star \int_0^{\varepsilon_{j,k-j}} (-g_{n,k})'(q)\, \mathrm{d}q,
\end{align*}
and for $t>k-j$, we have
\begin{align*}
    &\int_0^1 \alpha_{t+1,j}(q)\, \mathrm{d}q + \sum_{\tau \leq t} \int_0^1 q \alpha_{\tau,j}(q)\, \mathrm{d}q \\
    &\leq B_j \frac{\ln(\barnk)}{\barnk} + \theta_j^\star\left( \int_{\varepsilon_{j,t}}^{\varepsilon_{j,t+1}} (-g_{n,k})'(q)\, \mathrm{d}q + \int_0^{\varepsilon_{j,t}} q (- g_{n,k})'(q)\, \mathrm{d}q  \right).
\end{align*}
Then, for any $t\geq k-j$, we obtain the inequality
\begin{align}
    & \int_0^1 \alpha_{t+1,j}(q)\, \mathrm{d}q + \sum_{\tau \leq t} \int_0^1 q \alpha_{\tau,j}(q)\, \mathrm{d}q \nonumber\\
    & \leq B_j \frac{\ln (\barnk)}{\barnk} + \theta_j^\star\left( \int_{\varepsilon_{j,t}}^{\varepsilon_{j,t+1}} (-g_{n,k})'(q)\, \mathrm{d}q + \int_0^{\varepsilon_{j,t}} q (- g_{n,k})'(q)\, \mathrm{d}q  \right)+\theta_j^\star \int_0^{\varepsilon_{j,k-j}} (-g_{n,k})'(q)\, \mathrm{d}q. \label{ineq:upper_bound_for_alpha_jt}
\end{align}
We upper bound separately the last two terms in~\eqref{ineq:upper_bound_for_alpha_jt}. For the first term, we have
\begin{align}
    &\left[\frac{n}{(k-1)!} \left( 1 + 16 \frac{k^2}{n} \right)\right]^{-1}\left(\int_{\varepsilon_{j,t}}^{\varepsilon_{j,t+1}} (-g_{n,k})'(u)\, \mathrm{d}u + \frac{k}{n+1}\int_0^{\varepsilon_{j,t}} (-g_{n+1,k+1})'(u)\, \mathrm{d}u \right)\notag\\
    &\le \Gamma_k(-\ln y_{j,t}) - \Gamma_{k}(-\ln y_{j,t+1}) + \frac{1}{\barnk} (k! - \Gamma_{k+1}(-\ln y_{j,t}))  \notag\\
    &= \frac{1}{\barnk}\left( \Gamma_{k}(-\ln y_j)'\left( \frac{t}{\barnk} \right) - \frac{\Gamma_{k}(-\ln y_{j,t+1}) -\Gamma_{k}(-\ln y_{j,t})}{1/\barnk} + \frac{\theta^{\star}_{j+1}}{\theta^{\star}_j}( k! - \Gamma_{k+1}(-\ln y_{j+1,t})) \right)\notag\\
    &=  \frac{1}{\barnk}\left( \Phi_{k,\ell}'(t/\barnk) - \frac{ \Phi_{k,\ell}((t+1)/\barnk)-\Phi_{k,\ell}(t/\barnk)}{1/\barnk} + \frac{\theta^{\star}_{j+1}}{\theta^{\star}_j}( k! - \Gamma_{k+1}(-\ln y_{j+1,t})) \right).\label{eq:gamma_kl_approx}
\end{align}
The first inequality follows by Proposition~\ref{prop:g-n-properties}. The following claim allows us to guarantee that $\varepsilon_{\ell,k}$ is close to zero, for all $\ell$, which allows us to use~\ref{prop:g-n-properties}\ref{prop:limit_props_of_gnk-2}. The proof simply uses Proposition~\ref{prop:mon-phi}\ref{prop:nls-lower_bound} for $\barnk \geq N_k$ and we skip it for brevity. 
\begin{claim}\label{claim:small_epsilon}
    For any $\ell$, we have $\varepsilon_{\ell,k}\leq 2 \bar{c}_k \ln(\barnk)^{1/k}/\barnk^{1+2/k}$, where $\bar{c}_k$ is defined in Proposition~\ref{prop:mon-phi}.
\end{claim}
Note that the claim implies that for $n$ large, $\varepsilon_{\ell,k}\leq (k-1)/\barnk \leq 2/\sqrt{\barnk}$. In addition to this claim, the following claims allow us to bound the terms in the parenthesis in~\eqref{eq:gamma_kl_approx}. We defer their proof to Appendix~\ref{app:LB}.
\begin{claim}\label{claim:phi-log} It holds that
$\barnk( \Phi_{k,\ell}(t/\barnk)-\Phi_{k,\ell}((t+1)/\barnk)) + \Phi_{k,\ell}'(t/\barnk)\le c_k\ln(\barnk)/\barnk$, where $c_k>0$ is defined in Proposition~\ref{prop:mon-phi}.
\end{claim}
\begin{claim}\label{claim:UB_for_terms_}
    For $n$ sufficiently large, we have
    \begin{align*}
        \left( 1- 4\frac{(k+1)^2}{n+1} \right)^{-1}\left( 1-\frac{\barnk \varepsilon_{j+1,t}^2}{1-\varepsilon_{j+1,t}}\right)^{-1}\leq 1 + 10 \frac{\ln (\barnk)^2}{\barnk}.
    \end{align*}
\end{claim}
Hence, we can further bound (\ref{eq:gamma_kl_approx}) as follows:
\begin{align*}
    &\leq \frac{1}{\barnk }\left( c_k \frac{\ln (\barnk)}{\barnk} + \frac{\theta^{\star}_{j+1}}{\theta^{\star}_j}( k! - \Gamma_{k+1}(-\ln y_{j+1,t})) \right) \tag{Using Claim~\ref{claim:phi-log}} \\
    &\leq \frac{1}{\barnk} \left( c_k \frac{\ln (\barnk)}{\barnk} + 
    \left(1-\frac{4(k+1)^2}{n+1}\right)^{-1}\left(1-\frac{\barnk\varepsilon_{j+1,t}^2}{1-\varepsilon_{j+1,t}}\right)^{-1}
    \frac{\theta^{\star}_{j+1}}{\theta^{\star}_j} (k-1)!\int_0^{\varepsilon_{j+1,t}}( -g_{n,k}'(u)) u \, \mathrm{d}u\right) \tag{Using~\ref{prop:g-n-properties}\ref{prop:limit_props_of_gnk-2} and Claim~\ref{claim:small_epsilon}} \\
    & \leq \frac{1}{\barnk} \left( c_k \frac{\ln (\barnk)}{\barnk} + 
    \left(1 + 10 \frac{\ln (\barnk)^2}{\barnk}\right) 
    \frac{\theta^{\star}_{j+1}}{\theta^{\star}_j} (k-1)!\int_0^{\varepsilon_{j+1,t}}( -g_{n,k}'(u)) u \, \mathrm{d}u\right)\tag{Using Claim~\ref{claim:UB_for_terms_}} \\
    & \leq \frac{1}{\theta_j^\star \barnk} \left( 1 + 10 \frac{\ln (\barnk)^2}{\barnk} \right) \left( \theta^{\star}_j c_k \frac{\ln \barnk}{\barnk} + \theta^{\star}_{j+1}(k-1)!\int_0^{\varepsilon_{j+1,t}} (-g_{n,k}'(u)) u \, \mathrm{d}u\right). 
\end{align*}
From here, we obtain
\begin{align*}
&\theta_j^\star\left(\int_{\varepsilon_{j,t}}^{\varepsilon_{j+1,t}} (-g_{n,k})'(q)\, \mathrm{d}q + \int_0^{\varepsilon_{j,t}} q(-g_{n,k})'(q)\, \mathrm{d}q\right)\\
& \leq \left(1+12 \frac{\ln(\barnk)^2}{\barnk}\right) \frac{c_k}{k!} \frac{\ln(\barnk)}{\barnk} + \left( 1+ 12\frac{\ln (\barnk)^2}{\barnk}\right)\int_0^{\varepsilon_{j+1,t}} \theta_{j+1}^\star q(-g_{n,k})'(q)\, \mathrm{d}q \\
& \leq \left(1+12 \frac{\ln(\barnk)^2}{\barnk}\right) \frac{c_k}{k!} \frac{\ln(\barnk)}{\barnk} + \left( 1+ 12\frac{\ln (\barnk)^2}{\barnk}\right)\left( \sum_{\tau\leq t}\int_0^1 q \alpha_{\tau,j+1}(q)\, \mathrm{d}q - B_{j+1}\frac{\ln (\barnk)}{\barnk}   \right).
\end{align*}
We now bound the last term in~\eqref{ineq:upper_bound_for_alpha_jt}. Note that the function $-\Gamma_{k}(\barnk u)' = \barnk (\barnk u)^{k-1} e^{-\barnk u}$ is increasing in $[0,(k-1)/\barnk]$ and decreasing in $[(k-1)/\barnk,+\infty)$. Then,
\begin{align*}
    & \theta_j^\star\int_0^{\varepsilon_{j,k-j}}(-g_{n,k})'(q)\, \mathrm{d}q \\
    &\leq \frac{1}{k}\int_0^{\varepsilon_{j,k}} (-g_{n,k})'(q)\, \mathrm{d}q \tag{Since $\varepsilon_{j,k-j}\leq \varepsilon_{j,k}$ and $\theta_j^\star\leq 1/k$}\\
    & \leq \frac{n}{k!}\left(1 + 4 \frac{k^2}{n} \right)\int_0^{\varepsilon_{j,k}}(-\Gamma_k(\barnk u))' \, \mathrm{d}u \tag{Using Proposition~\ref{prop:g-n-properties}}\\
    &\leq \frac{n}{k!}\left( 1 + 4 \frac{k^2}{n} \right)  \barnk^{k} \left(   2 \bar{c}_k \frac{\ln(\barnk)^{1/k}}{\barnk^{1+2/k}} \right)^k  \tag{Using Claim~\ref{claim:small_epsilon}} \\
    & \leq \frac{2^k k^{1/k}}{k!}  c_k^k \left( 1 + 8 \frac{k^2}{n} \right) \frac{\ln(\barnk)}{\barnk}   \\ 
    & \leq 4 c_k^k \left( 1+ 12 \frac{\ln (\barnk)^2}{\barnk} \right)\frac{\ln (\barnk)}{\barnk},
\end{align*}
where we used that $2^k k^{1/k}\leq 4 k!$ for all $k\geq 1$ and the bound $8k^2/n \leq 12 \ln(\barnk)^2/\barnk$ for $n$ large. Then,
\begin{align*}
    &\int_0^1 \alpha_{t+1,j}(q) \, \mathrm{d}q + \sum_{\tau \leq t}\int_0^1 q \alpha_{\tau,j}(q)\, \mathrm{d}q\\
    & \leq \left( 1+ 12 \frac{\ln(\barnk)^2}{\barnk} \right)\sum_{\tau \leq t} \int_0^1 q\alpha_{\tau,j+1}(q)\, \mathrm{d}q \\
    & \qquad + 4 c_k^k\left( 1 + 8 \frac{k^2}{n} \right)  \frac{\ln(\barnk)}{\barnk} + B_j \frac{\ln (\barnk)}{\barnk} + \left( 1 + 12 \frac{\ln(\barnk)^2}{\barnk} \right) \frac{c_k}{k!}\frac{\ln (\barnk)}{\barnk} \\
    & \qquad - \left( 1 + 12 \frac{\ln(\barnk)^2}{\barnk} \right) \frac{\ln(\barnk)}{\barnk} B_{j+1}\\
    & \leq \left(  1 + 12\frac{\ln (\barnk)^2}{\barnk}  \right) \sum_{\tau \leq t}\int_0^1 q \alpha_{\tau,j+1}(q)\, \mathrm{d}q +  \left(  1+ 12 \frac{\ln(\barnk)^2}{\barnk}  \right)\frac{\ln(\barnk)}{\barnk}\left( 4c_k^k + \frac{c_k}{k!} + B_j -B_{j+1}   \right) \\
    & = \left(  1  + 12 \frac{\ln(\barnk)^2}{\barnk}\right) \sum_{\tau\leq t}\int_0^1 q \alpha_{\tau,j+1}(q)\, \mathrm{d}q,
\end{align*}
where we used that $B_j=(4c_k^k + c_k/k!) \cdot (j-1)$ for $j\geq 1$. This finishes the proof of Claim~\ref{claim:bound_for_alpha_t}.
\end{proof}

\color{black}


\section{A Tight Prophet Inequality for Sequential Assignment}\label{sec:ssap}

In this section, we show that our new provable lower bounds for the $k$-selection prophet inequality imply a tight approximation ratio for the i.i.d. sequential stochastic assignment problem by \citet{derman1972sequential}, that we call $\ssap$ in what follows.
We provide the proof in two steps. 
Firstly, we show that $\ssap$ is more general than $\kipi{}$ in the sense that any policy for the sequential stochastic assignment problem with $n$ time periods implies a policy for $\kipi{}$ for any $k\in [n]$ (Proposition~\ref{prop:upper_bound_ssap}). 
This shows that the approximation ratio cannot be larger than $\min_{k\in[n]}\gamma_{n,k}$. Secondly, we match the upper bound by using the structure of the optimal policy for $\ssap$ (Proposition~\ref{prop:lower_bound_ssap}).

In the \ssap, the input is given by $n$ non-negative values (rewards) $r_1\leq r_2 \leq \cdots \leq r_n$ and we observe exactly $n$ non-negative values, presented one after the other in $n$ time periods, and drawn independently from a distribution $F$. 
{For notational convenience, we assume that time starts at $t=n$ and decreases all the way down to $t=1$, i.e., the value $t$ represents the number of time periods that remain before the next value is presented.}
For every period $t$, we observe the value $X_t\sim F$, and we have to irrevocably assign the value $X_t$ to one of the unassigned rewards $r_\tau$'s. 
The goal is to find a sequential policy $\pi$ that maximizes $v_{n,F,r}(\pi)=\EE\left[\sum_{t=1}^n X_t r_{\pi(t)}\right]$ where $\pi$ is a permutation of $[n]$. Note that the optimal offline value corresponds to $\sum_{t=1}^n r_t \EE\left[ X_{(t)} \right]$.
We denote by $\alpha_n$ the largest approximation ratio that any policy can attain in $\ssap$ for instances with $n$ time periods.

\begin{proposition}\label{prop:upper_bound_ssap}
    For every $n$, it holds that $\alpha_n\leq \min_{k\in [n]}\gamma_{n,k}$.
\end{proposition}

\begin{proof}
    Let $\pi$ be a policy for $\ssap$ with approximation ratio $\alpha$. 
    Given $k\in [n]$, we use $\pi$ to construct a policy $\pi'$ for $\kipi{}$. Without loss of generality, we can assume that $\OPT_{n,k}=1$. Fix $\varepsilon\in (0,1/n^2)$ and consider the following instance for $\ssap$: $r_i=\varepsilon i$ for each $i\in \{1,\ldots, n-k\}$ and $r_i=1$ for each $i\in \{n-k+1,\ldots,n\}$. 
    The policy $\pi'$ simulates $\pi$ by creating $r_1\leq r_2 \leq \cdots \leq r_n$ as defined before. When $\pi$ assigns $X_t$ to some $r_i=1$, then $\pi'$ selects the value $X_t$, while if $\pi$ assigns $X_t$ to some $r_i=\varepsilon i$, then $\pi'$ discards $X_t$. Then, we have that
    \begin{align*}
        \EE\left[\sum_{t=1}^n X_t \mathbf{1}_{\{t \text{ selected by }\pi\}}\right] 
        & \geq \alpha\EE\left[ \sum_{t=1}^n r_t X_{(t)}  \right] - \varepsilon n\\
        & \geq \alpha \EE\left[ \sum_{t=1}^k X_{(n-t+1)}  \right] - (1+\alpha)\varepsilon n = \alpha - (1+\alpha)\varepsilon n.
    \end{align*}
    This shows that the approximation ratio of $\pi'$ is at least $\alpha- (1+\alpha)\varepsilon n$. Since this holds for any $\varepsilon\in (0,1/n^2)$, we conclude that $\pi'$ has an approximation ratio of at least $\alpha$. 
    Since this holds for any $\ssap{}$ policy for $n$ time periods and any $k\in [n]$, we conclude the proof. 
\end{proof}

\begin{proposition}\label{prop:lower_bound_ssap}
For every $n$, it holds that $\alpha_n\geq \min_{k\in [n]}\gamma_{n,k}$.
\end{proposition}

To prove this proposition, we need to use the structure of the optimal dynamic programming policy for $\ssap$ shown by~\citet{derman1972sequential}, which we describe in what follows.
For each time distribution $F$ and each time $t$, there exist values $0=\mu_{0,t}(F) \leq \mu_{1,t}(F) \leq \cdots \leq \mu_{t,t}(F)$, where the value $\mu_{i,t}$ is the optimal expected value in problem with $t-1$ time periods in which the reward $r_i$ is assigned under the optimal policy.
If $X_t \in [\mu_{\tau-1,t}(F), \mu_{\tau,t}(F)]$ then, the optimal policy assigns $X_t$ with the $\tau$-th smallest available reward for $\tau\in \{1,\ldots,t\}$.
Furthermore, \citet{derman1972sequential} show that $v_{n,F,r}(\pi^\star)=\sum_{t=1}^n r_t \mu_{t,n+1}(F)$.
Note that the values $\mu$ are completely independent of the rewards, and they just depend on the distribution $F$ and $n$. 
\begin{proof}[Proof of Proposition \ref{prop:lower_bound_ssap}]
For every $\ell\in [n]$, let $d_{\ell}=r_{\ell}-r_{\ell-1}$ where $r_0=0$.
Since the rewards $r_t$ are non-decreasing in $t$, we have $d_{\ell}\ge 0$ for every $\ell\in [n]$, and $r_j=\sum_{\ell=1}^j d_{\ell}$.
Then, for every distribution $F$, we have
\begin{align}
\frac{\sum_{t=1}^n r_t \mu_{t,n+1}(F)}{\sum_{t=1}^n r_t \EE[ X_{(t)}]} & = \frac{\sum_{\tau=1}^n d_\tau \sum_{t=\tau}^{n}  \mu_{t,n+1}(F) }{\sum_{\tau=1}^n d_\tau \sum_{t=\tau}^n \EE[X_{(t)}]} \geq \min_{\tau\in [n]} \frac{\sum_{t=\tau}^n \mu_{t,n+1}(F)}{\sum_{t=\tau}^n \EE[X_{(t)}]}.\notag
\end{align}
Note that $\sum_{t=\tau}^n \mu_{t,n+1}(F)$ is the reward collected by the optimal policy $\pi^\star$ in the instance $r_1=\cdots =r_{\tau-1}=0 < 1 = r_{\tau}=\cdots = r_n$. 
Furthermore, $\sum_{t=\tau}^n\EE[X_{(t)}]$ is the sum of the $n-\tau+1$ largest values in a sequence of $n$ i.i.d. samples from $F$, i.e., $\sum_{t=\tau}^n\EE[X_{(t)}]=\OPT_{n,n-\tau+1}(F)$.
Therefore, the ratio inside the minimization operator can be interpreted as the ratio in a $\kipi{}$ with $k=n-\tau+1$. Since $\pi^\star$ is optimal for the instance $r$ described above, then it must be the case that $v_{n,F,r}(\pi^\star)=\sum_{t=\tau}^n \mu_{t,n+1}(F)\geq \gamma_{n,n-\tau+1} \OPT_{n,n-\tau+1}(F)$. 
The proof follows since this holds for every $\tau\in [n]$.
\end{proof}

Proposition \ref{prop:upper_bound_ssap} and Proposition \ref{prop:lower_bound_ssap} imply that $\alpha_{n}=\min_{k\in [n]}\gamma_{n,k}$ for every $n$.
The $1-k^ke^{-k}/k!$ lower bound on $\gamma_{n,k}$ imply that $\gamma_{n,k}$ is at least 0.78 for $k\geq 3$ (see, e.g., \citep{Duetting2020,beyhaghi2021improved}) which is in particular larger than $\liminf_n\gamma_{n,1}\approx 0.745$. 
Since our results imply that $\liminf_n\gamma_{n,2}\ge 0.829$, we conclude that $\liminf_n\alpha_n=\liminf_n\gamma_{n,1}\approx 0.745$.

\section{Final Remarks}\label{sec:final_remarks}

In this work, we provide a new exact formulation for $\kipi$. From our formulation, we can derive the nonlinear system of differential equations \eqref{ode1}-\eqref{ode3} as $n$ tends to infinity. Using this system, we can obtain provable guarantees for the approximation ratio of $\kipi$ when $n$ is large. We use the nonlinear system~\eqref{ode1}-\eqref{ode3} to provide new improved bounds for small values of $k \in \{2, \ldots, 5\}$. As a direct application of our new bounds, we provide a tight approximation ratio for the $\ssap$.

We also remark that our linear programming formulation offers several characteristics that make it suitable for a nonlinear analysis. \cite{jiang2022tightness} use a different formulation to provide bounds on the $k$-selection problem, and for the particular case of $k=1$, they also connect to the Hill and Kertz equation. Nevertheless, they go through an extra intermediate formulation in their limit analysis. With our approach, we can directly provide a feasible solution in quantile space that converges to a solution of the nonlinear system (see Subsection~\ref{subsec:warm-up}), and furthermore, it works for general values of $k$. 
To complement our new lower bounds provided by Theorem \ref{thm:comp-ratio}, we further provide new upper bounds in Table~\ref{tab:table_2} for values of $k\in \{1,\ldots,5\}$ obtained by constructing a finite-dimensional linear program over a particular grid of $[0,1]$ whose optimal value $\gamma_{n,k}'$ upper-bounds the optimal value of \ref{form:LP_dual} which is precisely the worst-case approximation ratio $\gamma_{n,k}$; we show the details of this construction and the numerical experiments in Appendix~\ref{app:numerical_k_2}.

\begin{table}[h!]
        \centering
        \begin{tabular}{|c|c|c|c|c|c|}
        \hline
            $k$ & 1 & 2 & 3 & 4 & 5  \\
            \hline\hline
            $\gamma_{n,k}'$ ($n=1000$) & 0.7474 & 0.8372 & 0.8742 & 0.8949 & 0.9086  \\
            \hline
            \cite{jiang2022tightness} ($n=1000$) & 0.7475 & 0.8377 & 0.8748 & 0.8955 & 0.9093 \\
            \hline
        \end{tabular}
        \caption{New numerical upper bounds on the worst-case approximation ratio $\inf_{n\geq 1} \gamma_{n,k}$, for $k\in \{1,2,\ldots,5\}$. We remark that \emph{smaller} values are better for upper bounds. Our grid for $n=1000$ produces smaller values than the state of the art.}
        \label{tab:table_2}
\end{table}

Finding an analytical formula for the approximation ratio of $\kipi$ remains an open problem, but our findings offer a potential avenue toward this goal. Using the system~\eqref{ode1}-\eqref{ode3}, we can characterize the value $\theta_k^\star$ using an integral equation and show that $\theta_k^\star \geq (1 - e^{-k})/k$. Providing a similar lower bound for $j < k$ becomes nontrivial due to the dependency with $j + 1$. 

We also remark that, in principle, our solution is suboptimal as we only construct a feasible solution to the weak dual~\ref{form:dual-relaxation}. Another place where suboptimality could appear is in the the weak duality between~\ref{form:dual-relaxation} and~\ref{form:LP_dual}. However, we believe these two programs hold strong duality. In fact, the following LP:
\begin{align}
    \sup \quad  v \hspace{1cm} &\tag*{\normalfont{\mbox{[D$_{\text{strong}}$]$_{n,k}$}}}\label{form:strong_dual} \\
	\normalfont\text{s.t.} \quad  \int_0^1 \beta_{1,\ell}(q)\, \mathrm{d}q &\leq \mathbf{1}_{k}(\ell),  \hspace{.3cm} \text{for all }\ell\in [k],\nonumber\\
	\int_0^1 \beta_{t+1,k}(q) \, \mathrm{d}q&\leq \int_0^1 (1-q) \beta_{t,k}(q) \, \mathrm{d}q,\hspace{.3cm}   \text{for all } t\in [n-1] ,\nonumber\\
	\int_0^1 \beta_{t+1,\ell}(q) \, \mathrm{d}q&\leq \int_0^1 (1-q) \beta_{t,\ell}(q) \, \mathrm{d}q+ \int_0^1 q \beta_{t,\ell+1}(q) \, \mathrm{d}q,   \hspace{.2cm}\text{all } t\in [n-1], \ell\in[k-1],\nonumber\\
	 v g_{n,k}(u) + \frac{\mathrm{d}\eta}{\mathrm{d} u}(u)&\leq \sum_{t=1}^n\sum_{\ell=1}^k \int_u^1 \beta_{t,\ell}(q) \, \mathrm{d}q,  \text{ for }u \in [0,1]  , \nonumber \\
	\beta_{t,\ell}(q)&\ge 0\hspace{.3cm} \text{ for all }q\in [0,1],t\in [n]\text{ and }\ell\in [k],\nonumber \\
    \eta(q) & \geq 0 \hspace{0.3cm} \text{ for all } q\in [0,1], \nonumber \\
    \eta(0) & = \eta(1) = 0, \nonumber
\end{align}
can be shown to be a strong dual to~\ref{form:LP_dual}. The weak duality proof is analogous to the proof of Lemma~\ref{lem:weak_duality}, and the strong duality holds using a discretization argument as in~\citep{perez2022iid}. We remark that~\ref{form:dual-relaxation} is a restriction of~\ref{form:strong_dual} with $\eta=0$.

Our results are valid for $n$ large enough ($n \geq n_0$, with $n_0$ depending only on $k$); hence, providing a lower bound for all $n$ remains an open problem and will require additional structural results over the solution of the nonlinear system~\eqref{ode1}-\eqref{ode3}. For instance, for $k>1$ and $j<k$, we observe that the higher-order derivatives of $y_j$ change signs over $[0,1]$, hence, ruling out techniques that work in the case $k=1$ (see, e.g.,~\citep{correa2021posted,brustle2022competition}). We leave the tightness of our approximation result in Theorem \ref{thm:comp-ratio} through our nonlinear system~\eqref{ode1}-\eqref{ode3} as an open question, i.e., whether $\inf_{n\geq 1} \gamma_{n,k} = \sum_{\ell=1}^k \theta_\ell^\star$.


\bibliographystyle{abbrvnat}
{\small \bibliography{refs}}

\begin{thebibliography}{39}
\providecommand{\natexlab}[1]{#1}
\providecommand{\url}[1]{\texttt{#1}}
\expandafter\ifx\csname urlstyle\endcsname\relax
  \providecommand{\doi}[1]{doi: #1}\else
  \providecommand{\doi}{doi: \begingroup \urlstyle{rm}\Url}\fi

\bibitem[Alaei(2014)]{Alaei2011}
S.~Alaei.
\newblock Bayesian combinatorial auctions: Expanding single buyer mechanisms to
  many buyers.
\newblock \emph{SIAM Journal on Computing}, 43\penalty0 (2):\penalty0 930--972,
  2014.

\bibitem[Arnosti and Ma(2023)]{arnosti2023tight}
N.~Arnosti and W.~Ma.
\newblock Tight guarantees for static threshold policies in the prophet
  secretary problem.
\newblock \emph{Operations research}, 71\penalty0 (5):\penalty0 1777--1788,
  2023.

\bibitem[Assaf et~al.(2002)Assaf, Goldstein, and Samuel-Cahn]{assaf2002ratio}
D.~Assaf, L.~Goldstein, and E.~Samuel-Cahn.
\newblock Ratio prophet inequalities when the mortal has several choices.
\newblock \emph{The Annals of Applied Probability}, 12\penalty0 (3):\penalty0
  972--984, 2002.

\bibitem[Babaioff et~al.(2007)Babaioff, Immorlica, Kempe, and
  Kleinberg]{babaioff2007knapsack}
M.~Babaioff, N.~Immorlica, D.~Kempe, and R.~Kleinberg.
\newblock A knapsack secretary problem with applications.
\newblock In \emph{Approximation, randomization, and combinatorial
  optimization. Algorithms and techniques}, pages 16--28. 2007.

\bibitem[Beyhaghi et~al.(2021)Beyhaghi, Golrezaei, Leme, P{\'a}l, and
  Sivan]{beyhaghi2021improved}
H.~Beyhaghi, N.~Golrezaei, R.~P. Leme, M.~P{\'a}l, and B.~Sivan.
\newblock Improved revenue bounds for posted-price and second-price mechanisms.
\newblock \emph{Operations Research}, 69\penalty0 (6):\penalty0 1805--1822,
  2021.

\bibitem[Brustle et~al.(2024)Brustle, Correa, Duetting, and
  Verdugo]{brustle2022competition}
J.~Brustle, J.~Correa, P.~Duetting, and V.~Verdugo.
\newblock The competition complexity of dynamic pricing.
\newblock \emph{Mathematics of Operations Research}, 49\penalty0 (3):\penalty0
  1986--2008, 2024.

\bibitem[Buchbinder et~al.(2014)Buchbinder, Jain, and
  Singh]{buchbinder2014secretary}
N.~Buchbinder, K.~Jain, and M.~Singh.
\newblock Secretary problems via linear programming.
\newblock \emph{Mathematics of Operations Research}, 39\penalty0 (1):\penalty0
  190--206, 2014.

\bibitem[Chakraborty et~al.(2010)Chakraborty, Even-Dar, Guha, Mansour, and
  Muthukrishnan]{chakraborty2010approximation}
T.~Chakraborty, E.~Even-Dar, S.~Guha, Y.~Mansour, and S.~Muthukrishnan.
\newblock Approximation schemes for sequential posted pricing in multi-unit
  auctions.
\newblock In \emph{WINE 2010}, pages 158--169, 2010.

\bibitem[Chan et~al.(2014)Chan, Chen, and Jiang]{chan2014revealing}
T.~H. Chan, F.~Chen, and S.~H.-C. Jiang.
\newblock Revealing optimal thresholds for generalized secretary problem via
  continuous lp: impacts on online k-item auction and bipartite k-matching with
  random arrival order.
\newblock In \emph{SODA 2014}, pages 1169--1188, 2014.

\bibitem[Chawla et~al.(2010)Chawla, Hartline, Malec, and Sivan]{Chawla2010}
S.~Chawla, J.~D. Hartline, D.~L. Malec, and B.~Sivan.
\newblock Multi-parameter mechanism design and sequential posted pricing.
\newblock In \emph{STOC 2010}, page 311–320, 2010.

\bibitem[Correa and Cristi(2023)]{Correa2023combinatorial}
J.~Correa and A.~Cristi.
\newblock A constant factor prophet inequality for online combinatorial
  auctions.
\newblock In \emph{STOC 2023}, page 686–697, 2023.

\bibitem[Correa et~al.(2019{\natexlab{a}})Correa, D\"{u}tting, Fischer, and
  Schewior]{Correa2018}
J.~Correa, P.~D\"{u}tting, F.~Fischer, and K.~Schewior.
\newblock Prophet inequalities for i.i.d. random variables from an unknown
  distribution.
\newblock In \emph{EC 2019}, page 3–17, 2019{\natexlab{a}}.

\bibitem[Correa et~al.(2019{\natexlab{b}})Correa, Foncea, Hoeksma, Oosterwijk,
  and Vredeveld]{CorreaSurvey}
J.~Correa, P.~Foncea, R.~Hoeksma, T.~Oosterwijk, and T.~Vredeveld.
\newblock Recent developments in prophet inequalities.
\newblock \emph{SIGecom Exch.}, 17\penalty0 (1):\penalty0 61–70, may
  2019{\natexlab{b}}.

\bibitem[Correa et~al.(2019{\natexlab{c}})Correa, Foncea, Pizarro, and
  Verdugo]{Correa2019PricingtoProphets}
J.~Correa, P.~Foncea, D.~Pizarro, and V.~Verdugo.
\newblock From pricing to prophets, and back!
\newblock \emph{Operations Research Letters}, 47\penalty0 (1):\penalty0 25--29,
  2019{\natexlab{c}}.

\bibitem[Correa et~al.(2021)Correa, Foncea, Hoeksma, Oosterwijk, and
  Vredeveld]{correa2021posted}
J.~Correa, P.~Foncea, R.~Hoeksma, T.~Oosterwijk, and T.~Vredeveld.
\newblock Posted price mechanisms and optimal threshold strategies for random
  arrivals.
\newblock \emph{Mathematics of operations research}, 46\penalty0 (4):\penalty0
  1452--1478, 2021.

\bibitem[Derman et~al.(1972)Derman, Lieberman, and Ross]{derman1972sequential}
C.~Derman, G.~J. Lieberman, and S.~M. Ross.
\newblock A sequential stochastic assignment problem.
\newblock \emph{Management Science}, 18\penalty0 (7):\penalty0 349--355, 1972.

\bibitem[D\"{u}tting et~al.(2020)D\"{u}tting, Feldman, Kesselheim, and
  Lucier]{Duetting2020}
P.~D\"{u}tting, M.~Feldman, T.~Kesselheim, and B.~Lucier.
\newblock Prophet inequalities made easy: Stochastic optimization by pricing
  nonstochastic inputs.
\newblock \emph{SIAM Journal on Computing}, 49\penalty0 (3):\penalty0 540--582,
  2020.

\bibitem[Ehsani et~al.(2018)Ehsani, Hajiaghayi, Kesselheim, and
  Singla]{ehsani2018prophet}
S.~Ehsani, M.~Hajiaghayi, T.~Kesselheim, and S.~Singla.
\newblock Prophet secretary for combinatorial auctions and matroids.
\newblock In \emph{SODA 2018}, pages 700--714, 2018.

\bibitem[Ekbatani et~al.(2024)Ekbatani, Niazadeh, Nuti, and
  Vondr{\'a}k]{ekbatani2024prophet}
F.~Ekbatani, R.~Niazadeh, P.~Nuti, and J.~Vondr{\'a}k.
\newblock Prophet inequalities with cancellation costs.
\newblock In \emph{Proceedings of the 56th Annual ACM Symposium on Theory of
  Computing}, pages 1247--1258, 2024.

\bibitem[Feldman et~al.(2016)Feldman, Svensson, and
  Zenklusen]{feldman2016online}
M.~Feldman, O.~Svensson, and R.~Zenklusen.
\newblock Online contention resolution schemes.
\newblock In \emph{SODA 2016}, pages 1014--1033, 2016.

\bibitem[Goyal and Udwani(2023)]{goyal2019online}
V.~Goyal and R.~Udwani.
\newblock Online matching with stochastic rewards: Optimal competitive ratio
  via path-based formulation.
\newblock \emph{Operations Research}, 71\penalty0 (2):\penalty0 563--580, 2023.

\bibitem[Hajiaghayi et~al.(2007)Hajiaghayi, Kleinberg, and
  Sandholm]{Hajiaghayi2007}
M.~T. Hajiaghayi, R.~Kleinberg, and T.~Sandholm.
\newblock Automated online mechanism design and prophet inequalities.
\newblock In \emph{AAAI 2007}, page 58–65, 2007.

\bibitem[Hill and Kertz(1982)]{HillKertzStopRule}
T.~P. Hill and R.~P. Kertz.
\newblock {Comparisons of Stop Rule and Supremum Expectations of I.I.D. Random
  Variables}.
\newblock \emph{The Annals of Probability}, 10\penalty0 (2):\penalty0 336 --
  345, 1982.

\bibitem[Jiang et~al.(2023)Jiang, Ma, and Zhang]{jiang2022tightness}
J.~Jiang, W.~Ma, and J.~Zhang.
\newblock Tightness without counterexamples: A new approach and new results for
  prophet inequalities.
\newblock In \emph{EC 2023}, page 909, 2023.

\bibitem[Kertz(1986)]{kertz1986stop}
R.~P. Kertz.
\newblock Stop rule and supremum expectations of iid random variables: a
  complete comparison by conjugate duality.
\newblock \emph{Journal of multivariate analysis}, 19\penalty0 (1):\penalty0
  88--112, 1986.

\bibitem[Kesselheim et~al.(2014)Kesselheim, T{\"o}nnis, Radke, and
  V{\"o}cking]{kesselheim2014primal}
T.~Kesselheim, A.~T{\"o}nnis, K.~Radke, and B.~V{\"o}cking.
\newblock Primal beats dual on online packing lps in the random-order model.
\newblock In \emph{STOC 2014}, pages 303--312, 2014.

\bibitem[Kleinberg and Weinberg(2012)]{Kleinberg2012}
R.~Kleinberg and S.~M. Weinberg.
\newblock Matroid prophet inequalities.
\newblock In \emph{STOC 2012}, page 123–136, 2012.

\bibitem[Kolmogorov and Fomin(1975)]{kolmogorov1975introductory}
A.~N. Kolmogorov and S.~V. Fomin.
\newblock \emph{Introductory real analysis}.
\newblock 1975.

\bibitem[Krengel and Sucheston(1977)]{krengel1977semiamarts}
U.~Krengel and L.~Sucheston.
\newblock Semiamarts and finite values.
\newblock \emph{Bulletin of the American Mathematical Society}, 83\penalty0
  (4):\penalty0 745--747, 1977.

\bibitem[Lee and Singla(2018)]{lee2018optimal}
E.~Lee and S.~Singla.
\newblock Optimal online contention resolution schemes via ex-ante prophet
  inequalities.
\newblock In \emph{ESA 2018}, 2018.

\bibitem[Liu et~al.(2021)Liu, Leme, P\'{a}l, Schneider, and Sivan]{Liu2021}
A.~Liu, R.~P. Leme, M.~P\'{a}l, J.~Schneider, and B.~Sivan.
\newblock Variable decomposition for prophet inequalities and optimal ordering.
\newblock In \emph{EC 2021}, page 692, 2021.

\bibitem[Lucier(2017)]{lucier2017economic}
B.~Lucier.
\newblock An economic view of prophet inequalities.
\newblock \emph{ACM SIGecom Exchanges}, 16\penalty0 (1):\penalty0 24--47, 2017.

\bibitem[Mehta et~al.(2007)Mehta, Saberi, Vazirani, and
  Vazirani]{mehta2007adwords}
A.~Mehta, A.~Saberi, U.~Vazirani, and V.~Vazirani.
\newblock Adwords and generalized online matching.
\newblock \emph{Journal of the ACM}, 54\penalty0 (5):\penalty0 22--es, 2007.

\bibitem[Molina et~al.(2025)Molina, Gast, Loiseau, and
  Perchet]{molina2025prophet}
M.~Molina, N.~Gast, P.~Loiseau, and V.~Perchet.
\newblock Prophet inequalities: Competing with the top $\ell$ items is easy.
\newblock In \emph{SODA}, pages 1270--1307, 2025.

\bibitem[Mucci(1973)]{mucci1973differential}
A.~G. Mucci.
\newblock Differential equations and optimal choice problems.
\newblock \emph{The annals of Statistics}, pages 104--113, 1973.

\bibitem[Perez-Salazar et~al.(2024)Perez-Salazar, Singh, and
  Toriello]{perez2021robust}
S.~Perez-Salazar, M.~Singh, and A.~Toriello.
\newblock Robust online selection with uncertain offer acceptance.
\newblock \emph{Mathematics of Operations Research}, 2024.

\bibitem[Perez-Salazar et~al.(2025)Perez-Salazar, Singh, and
  Toriello]{perez2022iid}
S.~Perez-Salazar, M.~Singh, and A.~Toriello.
\newblock The iid prophet inequality with limited flexibility.
\newblock \emph{Mathematics of Operations Research}, 2025.

\bibitem[Samuel-Cahn(1984)]{SamuelCahn1984}
E.~Samuel-Cahn.
\newblock {Comparison of Threshold Stop Rules and Maximum for Independent
  Nonnegative Random Variables}.
\newblock \emph{The Annals of Probability}, 12\penalty0 (4):\penalty0 1213 --
  1216, 1984.

\bibitem[Yan(2011)]{yan2011mechanism}
Q.~Yan.
\newblock Mechanism design via correlation gap.
\newblock In \emph{SODA 2011}, pages 710--719, 2011.

\end{thebibliography}

\newpage
\appendix

\section{Missing Proof from Section \ref{sec:exact_LP_formulation}}\label{app:exact_LP_formulation}
\begin{proof}[Proof of Proposition \ref{prop:useful_LP}] 
We show that for any $j\in [n]$, 
$
\int_0^1 j{\binom{n}{j}}(1-u)^{j-1}u^{n-j}F^{-1}(1-u)\, \mathrm{d}u
= \EE[X_{(j)}].
$
This is sufficient since summing over all $j \in \{ n-k+1, \ldots, n\}$ will then complete the proof.
By performing a change of variables $x = F^{-1}(1-u)$, we get
\begin{align*}
&\int_0^1 j{\binom{n}{j}}(1-u)^{j-1}u^{n-j}F^{-1}(1-u)\, \mathrm{d}u\\
= & \int_{\infty}^0 j{\binom{n}{j}}(F(x))^{j-1}(1-F(x)))^{n-j}x(-f(x))\, \mathrm{d}x\\
= & \int_{0}^{\infty}\frac{n!}{(j-1)!(n-j)!}f(x)(F(x))^{j-1}(1-F(x)))^{n-j}x \, \mathrm{d}x = \EE[X_{(j)}],
\end{align*}
where $f(x) = F'(x)$. 
The final equality simply follows from the known fact that the probability density function 
$f_{X_{(j)}}(x) = n!f(x)(F(x))^{j-1}(1-F(x)))^{n-j}/((j-1)!(n-j)!).$
This finishes part \ref{prop:change_variable_opt}.

For part \ref{prop:change_variable_simple_exp}, recall that 
$\EE[X | X \geq x]{\Pr[X \geq x]} = \EE[X \mathds{1}_{\{X \geq x\}}]
$.
On the other hand, we have that
$$
\int_0^q F^{-1}(1-u)\, \mathrm{d}u =  \int_{\infty}^{F^{-1}(1-q)} z(-f(z)) \, \mathrm{d}z =  \int_x^{\infty} zf(z) \mathrm{d}z = \EE[X \mathds{1}_{\{X \geq x\}}],
$$
where we used the change of variable $z = F^{-1}(1-u)$, and in the second to last equality, we use that $q=\Pr[X\geq x] = 1-F(x)$.
This finished the proof.
\end{proof}

\section{Missing Proofs from Section~\ref{sec:LB}}\label{app:LB}

\begin{proof}[Proof of Claim \ref{claim-euler}]
    By induction $y_{m,j,t-1}>y_{m,j,t}$. Then, we can compare the following ratio
    \begin{align*}
        m\cdot \frac{\Gamma_{k}(-\ln y_{m,j,t-1})-\Gamma_{k}(-\ln y_{m,j,t}) }{\Gamma_{k+1}(-\ln y_{m,j,t-1})-\Gamma_{k+1}(-\ln y_{m,j,t})} & = m\cdot \frac{\int_{-\ln y_{m,j,t-1}}^{-\ln y_{m,j,t}} x^{k-1}e^{-x}\, \mathrm{d}x  }{\int_{-\ln y_{m,j,t-1}}^{-\ln y_{m,j,t}} x^{k}e^{-x}\, \mathrm{d}x } \\
        & \geq m \inf_{ x\in [-\ln y_{m,j,t-1}, -\ln y_{m,j,t}]  } \frac{1}{x} \\
        & = m \cdot \frac{1}{-\ln y_{m,j,t}}= \frac{m}{\ln m} \geq 1
    \end{align*}
    From here, the claim follows.
\end{proof}

\begin{proof}[Proof of Claim~\ref{claim:bounded_theta_finite_m}]
By contradiction, assume that $\theta_j > \theta_{j+1}^\star$. Let $t\leq t'$. Now, note that 
\begin{align*}
& \frac{1}{m} \left( k! - \Gamma_{k+1}(-\ln y_{m,j,t}) - \frac{\theta^{\star}_{j+1}}{\theta_j}(k!- \Gamma_{k+1}(-\ln Y_{j+1} (t/m) )) \right) \\
\geq & \frac{\theta^{\star}_{j+1}}{\theta_j m} \left( k! - \Gamma_{k+1}(-\ln y_{m,j,t}) - (k!- \Gamma_{k+1}(-\ln Y_{j+1} (t/m) )) \right)\\
= & - \frac{\theta^{\star}_{j+1}}{m\theta_j}\left( \Gamma_{k+1}(-\ln y_{m,j,t}) -  \Gamma_{k+1}(-\ln Y_{j+1} (t/m) )) \right)
\end{align*}
On the other hand, using~\eqref{eq:euler_approx_second_line} and~\eqref{eq:euler_approx_second_last_line} and Claim~\ref{claim-euler}, we obtain
\begin{align*}
 & \frac{1}{m} \left( k! - \Gamma_{k+1}(-\ln y_{m,j,t}) - \frac{\theta^{\star}_{j+1}}{\theta_j}(k!- \Gamma_{k+1}(-\ln Y_{j+1} (t/m) )) \right) \\ 
 \leq & - \frac{\theta^{\star}_{j+1}}{m\theta_j}( \Gamma_{k+1}( -\ln Y_{j+1}((t-1)/m)) -\Gamma_{k+1}(-\ln Y_{j+1}(t/m)))
\end{align*}
From here, we deduce $\Gamma_{k+1}( -\ln Y_{j+1}((t-1)/m)) \leq \Gamma_{k+1}(-\ln y_{m,j,t})$ or equivalently $Y_{j+1}((t-1)/m) \leq y_{m,j,t}$. For $t=1$, this last inequality implies $y_{n,j,1}\geq 1$ and this is impossible as we always have $y_{m,j,1}<1$ for any $\theta_j>0$. We conclude that $\theta_j \leq \theta_{j+1}^\star$.
\end{proof}

{ \begin{proof}[Proof of Proposition~\ref{prop:restricted_LP_tightens_constraints}]
In what follows, we use the convention that an empty sum is equal to zero. We also avoid writing the limits in the integrals and the differentials ``$\mathrm{d}q$'' as they are clear from the context. More specifically, we write $\int_0^1 h(q)\, \mathrm{d}q=\int h$ for any integrable function $h$ in $[0,1]$. We simply say that $(\alpha,v)$ is \emph{feasible} if $(\alpha,v)$ is feasible to~\ref{form:P_nk_restricted_n}. We also use the notation $\bar{n}=\barnk$ to avoid notational clutter. 

We fix $v\geq 0$ such that $(\alpha,v)$ is a feasible solution to~\ref{form:P_nk_restricted_n}. Let's define the sets
\begin{align*}
    J_k^{(\alpha,v)}&=\{ 
 t\in [\barnk]: (\alpha,v)\text{ does not tighten constraint~\eqref{const:constr_1_P_n_k_restricted_n}}\text{ for }t \}\\
 J_\ell^{(\alpha,v)} &= \{ 
 t\in [\barnk]: (\alpha,v)\text{ does not tighten constraint~\eqref{const:constr_t_P_n_k_restricted_n}}\text{ for }t, \ell\}, \quad \ell <k
\end{align*}
Let $t_{(\alpha,v)}' = \max\{ t\in J_1^{(\alpha,v)}\cup \cdots \cup J_k^{(\alpha,v)}  \}$. If 
$$J_1^{(\alpha,v)}\cup \cdots \cup J_k^{(\alpha,v)}=\emptyset,$$ then we define $t_{(\alpha,v)}'=\bar{n}+1$; otherwise, $t'\in [\bar{n}]$. Let $p_{(\alpha,v)}' = |\{ \ell \in [k]: t_{(\alpha,v)}'\in J_\ell^{(\alpha,v)} \}|$ be the number of constraints of type~\eqref{const:constr_1_P_n_k_restricted_n}-\eqref{const:constr_t_P_n_k_restricted_n} for which the $t_{(\alpha,v)}'$-th constraint is not tight.

Now, among all possible feasible solution $(\alpha,v)$, for a fixed $v$, choose the one that maximizes $t_{(\alpha,v)}'$. If $t_{(\alpha,v)}'=\bar{n}+1$, then we are done. Otherwise, let $t'\in [\bar{n}]$ be the maximum value for such a solution. Among all feasible solutions $(\alpha,v)$ such that $t_{(\alpha,v)}'=t'$ choose the one that minimizes $p'=p_{(\alpha,v)}'$. Note that $p'\geq 1$.
Now, we will modify $(\alpha,v)$ by finitely many mass transfers and additions yielding a new feasible solution $(\alpha',v)$ such that either $t_{(\alpha',v)}'>t'$ or either $t_{(\alpha',v)}'=t'$ and $p_{(\alpha',v)}' < p'$. In any case, we will obtain a contradiction.

Let's assume first that $t' \in J_k^{(\alpha,v)}$---the general case is handled similarly; we explain at the end the minor changes. We analyze two different cases:

\noindent{\bf Case 1.} If $t'=\bar{n}$, then, we consider the solution $\bar{\alpha}_{t,k}=\alpha_{t,k}$ for $t<\bar{n}$ and $\bar{\alpha}_{\bar{n},k}=\alpha_{\bar{n},k} + \varepsilon\mathbf{1}_{(0,1)}$ with $\varepsilon>0$ such that
    $\int \alpha_{\bar{n},k} + \varepsilon + \sum_{\tau< \bar{n}}\int q \alpha_{\tau,k} = 1.$
    Also, $\bar{\alpha}_{t,\ell}=\alpha_{t,\ell}$ for $\ell<k$. Note that $(\bar{\alpha},v)$ remains feasible and tightens one more constraint in~\eqref{const:constr_1_P_n_k_restricted_n}; this contradicts our choice of $p'$.\\
    
\noindent{\bf Case 2.} If $t'<n$, we define
    \[
    \bar{\alpha}_{t,k} = \begin{cases}
        \alpha_{t,k}, & t< t', \\
        \alpha_{t',k} + \sum_{\tau > t'} \omega_t \alpha_{t,k}, & t=t',\\
        (1-\omega_t)\alpha_{t,k}, & t>t',
    \end{cases}
    \]
    where $\omega_{t'+1},\ldots,\omega_n\in [0,1]$. Let $\bar{\alpha}_{t,\ell}=\alpha_{t,\ell}$ for $\ell<k$. 
    Note that $(\bar{\alpha},v)$ satisfies \eqref{const:P_n_approx_restricted_n}, it satisfies \eqref{const:constr_1_P_n_k_restricted_n} for $t<t'$, and for $t>t'$, we have
    $\int\bar{\alpha}_{t,k} +\sum_{\tau<t}\int q \bar{\alpha}_{\tau,k}= \int (1-\omega_t) \alpha_{t,k} + \sum_{\tau< t}\int q \alpha_{\tau,k} + \sum_{\tau\geq t} \omega_\tau \int q \alpha_{\tau,k}$, 
    which is increasing in $\omega_{\tau}$ for $\tau > t$ and decreasing in $\omega_{t}$.
    
    We start with the values $\omega_{t'+1},\ldots,\omega_{n}=0$ and at this point $(\bar{\alpha},v)$ is feasible. 
    By the choice of $t'$, we have
    $\int \alpha_{t',k} + \sum_{\tau> t'}\omega_{\tau}\int \alpha_{\tau,k} + \sum_{\tau< t}\int q \alpha_{\tau,k} \leq 1$
    for $\omega_{t'+1},\ldots,\omega_{n}>0$ small enough. Now, we increment $\omega_{t'+1}$ as much as possible while keeping feasibility of $(\bar{\alpha},v)$. We repeat the same process in the order $\omega_{t'+2},\ldots,\omega_n$. We note that $\omega_{t'+1},\ldots,\omega_{n}$ are not all $0$'s.
        
    Suppose that we have
        $\int\bar{\alpha}_{t',k} +\sum_{\tau<t'}\int q \bar{\alpha}_{\tau,k}=\int \alpha_{t',k} + \sum_{\tau> t'}\omega_{\tau}\int \alpha_{\tau,k} + \sum_{\tau< t}\int q \alpha_{\tau,k} < 1.$
        Then, we claim that $\omega_{t'+1},\ldots,\omega_{\bar{n}}=1$. Indeed, let $\tau'>t'$ be the smallest $\tau$ such that $\omega_{\tau}<1$. Then, 
        $\bar{\alpha}_{t,k} = 0, \text{ for }t\in \{t'+1,\ldots,\tau-1\}.$
        Furthermore, constraints~\eqref{const:constr_1_P_n_k_restricted_n} for $t\in \{t',t'+1,\ldots\tau-1\}$ are not tight, because they are dominated by constraint~\eqref{const:constr_1_P_n_k_restricted_n} for $t=t'$ which is not tight. Since increasing $\omega_{\tau}$ does not affect constraints~\eqref{const:constr_1_P_n_k_restricted_n} for $t\geq \tau$, we can increase slightly $\omega_{\tau}$ and contradict the choice of $\omega_{t'+1},\ldots,\omega_{\bar{n}}$. From this analysis, we also deduce that every constraint~\eqref{const:constr_1_P_n_k_restricted_n} for $t=t',\ldots,{\bar{n}}$ is not tight. Furthermore, $\bar{\alpha}_{t,k}=0$ for $t>t'$. Define
        \[
        \hat{\alpha}_{t,k}(q) = \begin{cases}
            \bar{\alpha}_{t,k}(q)\,(=\alpha_{t,k}(q)), & t< t',\\
            \bar{\alpha}_{t,k}(q) + c_{t}\delta_{\{ 1 \} }(q) , & t\geq t',
        \end{cases}
        \]
        where $\delta_{\{ 1\} }(\cdot)$ is the Dirac delta at one. We define $\hat{\alpha}_{t,\ell}=\bar{\alpha}_{t,\ell}=\alpha_{t,\ell}$ for $\ell< k$. 
        Note that $(\hat{\alpha},v)$ satisfies  constraints~\eqref{const:constr_t_P_n_k_restricted_n} and constraints~\eqref{const:P_n_approx_restricted_n}, and constraints~\eqref{const:constr_1_P_n_k_restricted_n} for $t<t'$. If we define
        $c_{t'}= 1- \int \bar{\alpha}_{t',k} - \sum_{\tau < t}\int q \bar{\alpha}_{\tau,k} >0$
        we have that $\hat{\alpha}$ satisfies constraint~\eqref{const:constr_1_P_n_k_restricted_n} at $t=t'$ with equality. For $t> t'$ we define
        $c_t = 1 - \sum_{\tau < t} \int q\hat{\alpha}_{\tau,k} \geq 0$. A small computation shows that $(\hat{\alpha},v)$ is again feasible and tightens \eqref{const:constr_1_P_n_k_restricted_n} for $t'$. Furthermore, all the other constraints~\eqref{const:constr_t_P_n_k_restricted_n} remain unchanged for $t\leq t'$ as they are only affected by terms $\alpha_{j,\tau}$ with $\tau < t'$. This implies that either $p_{(\hat{\alpha},v)}' < p'$ if $p'>1$ or $t_{(\hat{\alpha},v)}'> t'$ if $p'=1$. In any case, this leads again to a contradiction to our choice of $(\alpha,v)$.



When  $t'\notin J_{k}^{(\alpha,v)}$, we have $t'\in J_\ell^{(\alpha,v)}$ for some $\ell<k$. In this case, the analysis is the same with the only difference that the constraints will have the value $\sum_{\tau < t} \int q\alpha_{\tau,\ell+1}$ on the right-hand side instead of $1$. It is crucial to notice that this value is a non-negative constant when modifying $\alpha_{t,\ell}$; hence, our mass transfers and additions still work. We skip the details for brevity.
\end{proof}}


\begin{proof}[Proof of Proposition \ref{prop:g-n-properties}]
    Part \ref{gnk:diff-neg} follows directly by computing the derivative:
    \begin{align*}
        g_{n,k}'(u) &= \sum_{j=n-k+1}^n j \binom{n}{j} \left( -(j-1)(1-u)^{j-2}u^{n-j} + (n-j)(1-u)^{j-1}u^{n-j-1}   \right) \\
        & = \sum_{j=n-k+1}^{n-1} j \binom{n}{j} (n-j)(1-u)^{j-1}u^{n-j-1} - \sum_{j=n-k+1}^n j(j-1)\binom{n}{j} (1-u)^{j-2}u^{n-j} \\
        & = \sum_{j=n-k+2}^n (j-1)\binom{n}{j-1}(n-j+1)(1-u)^{j-2}u^{n-j} -  \sum_{j=n-k+1}^n j(j-1)\binom{n}{j} (1-u)^{j-2}u^{n-j} \\
        & = \sum_{j=n-k+2}^n \frac{n!}{(j-2)!(n-j)!}(1-u)^{j-2}u^{n-j} -\sum_{j=n-k+1}^n \frac{n!}{(j-2)!(n-j)!}(1-u)^{j-2}u^{n-j} \\
        & = - (n-k+1)(n-k)\binom{n}{k-1} (1-u)^{n-k-1}u^{k-1}.
    \end{align*}
    In the same line, part \ref{gnk:recursive} follows by evaluating directly $g'_{n+1,k+1}$ using the previous formula:
    \begin{align*}
    g_{n+1,k+1}'(u)&= - (n-k+1)(n-k)\binom{n+1}{k} (1-u)^{n-k-1}u^{k}=\frac{n+1}{k}ug'_{n,k}(u).
    \end{align*}
For part \ref{prop:limit_props_of_gnk}, we have
\begin{align*}
-g'_{n,k}(u) & = (n-k+1)(n-k)\binom{n}{k-1}(1-u)^{n-k-1}u^{k-1}\\
& = \frac{(n-k+1)}{\barnk^{k-1}}\binom{n}{k-1} (\barnk u)^{k-1}(1-u)^{n-k-1}(\barnk+1) \\
 & =\frac{(\barnk+2)}{\barnk^{k-1}}\frac{n\cdot(n-1) \cdots  (n-(k-1)+1)}{(k-1)!} (\barnk u)^{k-1}(1-u)^{n-k-1}(\barnk+1)\\
& \leq  \frac{n}{(k-1)!} \left( \frac{n-k/2}{\barnk} \right)^{k-1}(\barnk u)^{k-1}(1-u)^{n-k-1} (\barnk+1)\\
& \leq  \frac{n}{(k-1)!} \left( 1 + \frac{4k^2}{n} \right)(\barnk u)^{k-1}e^{-\barnk u} (\barnk+1).
\end{align*}
The final inequality follows by the bound $(1-u)^x \leq e^{-ux}$ for $u \in [0,1]$ and observing that 
\begin{align*}
    \left( \frac{n-k/2}{\barnk} \right)^{k-1}=  \left( 1+ \frac{k/2 +1}{n-k-1}\right)^{k-1}&\leq  \left( 1+ \frac{k+1}{n-(k+1)}\right)^{k+1}\\
    &\leq  \exp((k+1)^2/(n-(k+1)))
\end{align*}

Rewrite $n=(k+1)+c(k+1)^2$ for some $c>1$. We get 
$\exp((k+1)^2/(n-(k+1))) = \exp(1/c)$ and we can compute 
$\left( 1 + 4k^2/n \right) \geq 1+2/c$.
Thus the inequality holds when $\frac{1}{c} \leq \ln(1+2/c)$, which is true for any $c>2$.
That is, letting $n \geq (k+1)+2(k+1)^2$ suffices. 
This proves the claim as the above inequality can be slightly strengthened by a factor of $\barnk/(\barnk+1)$.

For \ref{prop:limit_props_of_gnk-2}, we have
    \begin{align*}
        -g'_{n,k}(u) & = (n-k+1)(n-k)\binom{n}{k-1}(1-u)^{n-k-1}u^{k-1}\\
        & \geq \frac{\barnk^2}{\barnk^{k-1}} \binom{n}{k-1} (\barnk u)^{k-1}(1-u)^{n-k-1} \\
        & = \frac{\barnk^2}{\barnk^{k-1}} \frac{n\cdot(n-1) \cdots  (n-(k-1)+1)}{(k-1)!} (\barnk u)^{k-1}(1-u)^{n-k-1} \\
        & \geq \frac{n \cdot \barnk}{(k-1)!} \left( 1 - \frac{k}{\barnk} \right)^{k-1}(\barnk u)^{k-1}(1-u)^{n-k-1} \\
        & \geq \frac{n \cdot \barnk}{(k-1)!} \left( 1 - 4 \frac{k}{n} \right)^k (\barnk u)^{k-1} e^{-\barnk u/(1-u)} \\
        & \geq \frac{n}{(k-1)!}\left( 1 - 4 \frac{k^2}{n} \right) \left( 1 - \frac{\barnk u^2}{1-u} \right) (\barnk u)^{k-1}e^{-\barnk u} \barnk 
    \end{align*}
    where in the third equality we use that $(1-u)^{-1}=1+u/(1-u)\leq \exp({u/(1-u)})$, and in last inequality we use $\exp({-\barnk u /(1-u)}) = \exp({-\barnk u -\barnk u^2/(1-u)}) \geq \exp({-\barnk u})(1-\barnk u^2/(1-u))$. Observe that $\Gamma_{k}(\barnk u)' = \Gamma_{k}'(\barnk u) \barnk = - (\barnk u)^{k-1} e^{-\barnk u} \barnk.$ 
    We conclude by noting that the function $1-\barnk x^2/(1-x)$, in $[0,1]$, is decreasing, positive at zero, and it has a unique root in the value $(\sqrt{4 \barnk + 1} - 1)/(2 \barnk)$.
    Since this value is larger than $1/(2\sqrt{\barnk})$, the conclusion follows.
\end{proof}

\begin{proof}[Proof of Proposition \ref{prop:mon-phi}]
Note that by \eqref{ode2-phi}, for every $t\in (0,1)$ we have 
$$|\Phi'_{k,\ell}(t)|\le k!+|\Phi_{k+1,\ell}(t)|+\frac{\theta_{\ell+1}^{\star}}{\theta_{\ell}^{\star}}k!+\frac{\theta_{\ell+1}^{\star}}{\theta_{\ell}^{\star}}|\Phi_{k+1,\ell+1}(t)|\le 4k!\frac{\theta_{\ell+1}^{\star}}{\theta_{\ell}^{\star}},$$
where the last inequality holds since $\Phi_{k+1,r}\le k!$ in $(0,1)$ for every $r\in [k]$, and $\theta_{\ell}^{\star}<\theta_{\ell+1}^{\star}$ by Proposition \ref{prop:nls}\ref{prop:nls-diff-theta}.
Let $b_k=4k!\max_\ell \theta^{\star}_{\ell+1}/\theta^{\star}_{\ell}$.
Then, since $\Phi_{k,\ell}(1)=0$, using the Taylor first-order approximation for $\Phi_{k,\ell}$ in one, for every $t\in (0,1)$ we have
$\Phi_{k,\ell}(t) \leq b_k (1-t).$
For each $\ell\in [k]$, by the formula $\Gamma_k(x)=(k-1)!\cdot e^{-x}\sum_{r=0}^{k-1}x^r/r!$ applied with 
$x=- \ln(Y_{\ell}(t))$
we conclude that
    $b_k (1-t)\ge \Phi_{k,\ell}(t)=(k-1)! \sum_{r=0}^{k-1}Y_{\ell}(t)(-\ln Y_{\ell}(t))^r/r!,$
    and then $Y_{\ell}(t)(-\ln Y_{\ell}(t))^r\le b_k\cdot r!/(k-1)!\le b_k,$
    where the first inequality holds since $Y_{\ell}(t)\in [0,1]$ for every $t\in (0,1)$.
    This concludes part \ref{prop:nls-b}.

For the second part, by \eqref{ode2-phi}, for each $\ell\ne k$ we have
    \begin{align*}
        -\Phi_{k,\ell}(t)  = \int_t^1 \Phi_{k,\ell}'(\tau) \, \mathrm{d}\tau &= \int_t^1\Big(k!-\Phi_{k+1,\ell}(\tau) - \frac{\theta^{\star}_{\ell+1}}{\theta^{\star}_{\ell}}(k!-\Phi_{k+1,\ell+1}(\tau))\Big)\, \mathrm{d}\tau\\
        & \leq (1-t)k!\left( 1 - \frac{\theta^{\star}_{\ell+1}}{\theta^{\star}_{\ell}}\right) + \frac{\theta^{\star}_{\ell+1}}{\theta^{\star}_{\ell}}\int_t^1 \Phi_{k+1,\ell+1}(\tau)\, \mathrm{d}\tau.
    \end{align*}
    For each $\ell\ne k$, choose $\delta_{\ell}>0$ such that 
    $\Phi_{k+1,\ell+1}(t) \leq k! (1-\theta^{\star}_{\ell}/\theta^{\star}_{\ell+1})/2$
    and $Y_{\ell}(t)^{1/2}(-\ln Y_k(t))^{k-1}\leq 1$ for $t\in (\delta_{\ell}, 1)$. 
    Using the first inequality, get
    \begin{align*}
    &(1-t)k!\left( 1 - \frac{\theta^{\star}_{\ell+1}}{\theta^{\star}_{\ell}}\right) + \frac{\theta^{\star}_{\ell+1}}{\theta^{\star}_{\ell}}\int_t^1 \Phi_{k+1,\ell+1}(\tau)\, \mathrm{d}\tau \\
    &\leq 
    (1-t)k!\left( 1 - \frac{\theta^{\star}_{\ell+1}}{\theta^{\star}_{\ell}}\right) + (1-t)k!\frac{1}{2}\left(\frac{\theta^{\star}_{\ell+1}}{\theta^{\star}_{\ell}} -1\right)= -k!(1-t)\frac{1}{2}\left(\frac{\theta^{\star}_{\ell+1}}{\theta^{\star}_{\ell}} -1\right).
    \end{align*}
    Hence for this interval, using the bound $\Phi_{k,\ell}(t) = \Gamma_k(-\ln Y_{\ell}(t)) \leq k! Y_{\ell}(t)(-\ln Y_{\ell}(t))^{k-1}$, 
    we have $Y_{\ell}(t)\geq (1-t)^2 (\theta^{\star}_{{\ell}+1}/\theta^{\star}_{\ell}-1)^2/4$.
    For ${\ell}=k$, we have
    \begin{align*}
        (1-t)k! \left( \frac{1}{k\theta_k} -1 \right) \leq \Gamma_k(-\ln Y_k(t)) \leq k! Y_k(t)(-\ln Y_k(t))^{k-1}.
    \end{align*}
    By part \ref{prop:nls-b}, we know that $Y_k(t)\leq b_k(1-t)$. Hence, for some $\delta_k>0$, $Y_k(t)^{1/2}(-\ln Y_k(t))^{k-1}\leq 1$ for all $t\in (\delta_k,1)$. Hence, $Y_k(t)\geq (1-t)^2 (1/k\theta_k -1)^2$.
 Part \ref{prop:nls-delta} follows by taking $\Delta_k=\max\{\delta_1,\ldots,\delta_k\}$ and $d_k=\min\{(\theta^{\star}_{j+1}/\theta^{\star}_j)-1: j\in \{1,\ldots,k-1\} \}/4$. Using Lemma~\ref{prop:nls}, we can conclude that $d_k > 0$.

From \eqref{ode1-phi} we have $\Phi''_{k,k}(t)=-\Phi'_{k+1,k}(t)\ge 0$, since $\Phi_{k+1,k}$ is non-increasing.
For $\ell\ne k$,
    \begin{align*}
        |\Phi_{k,\ell}''(t)| & = |-\Phi_{k+1,\ell}'(t) + \Phi_{k+1,\ell+1}'(t)| \\
        & \le |\Phi_{k,\ell}'(t)(-\ln Y_{\ell}(t))| + |\Phi_{k,\ell+1}'(t)(-\ln Y_{\ell+1}(t))| \\
        & \leq b_k \left( -\ln Y_{\ell}(t) - \ln Y_{\ell+1} (t) \right).
    \end{align*}
    Let $N_k=\max\{1/(1-\Delta_k)+1,1/d_k\}$.
    Then, for every $n\ge N_k$ we have $1-1/n>\Delta_k$. 
    By the previous part we have that $Y_{\ell}(1-1/n)\geq d_k n^{-2}$ for all $\ell$.
    Let $c_k=6b_k$.
    Since $-\ln Y_{\ell}(t) - \ln Y_{\ell+1} (t)$ is increasing as a function of $t$, for every $t\in (0,1-1/n)$ we have 
    $|\Phi''_{k,\ell}(t)|\le b_k\cdot 2\ln(n^2/d_k)\le b_k\cdot 2\ln(n^3)=c_k\ln(n)$.
    This concludes the proof of part \ref{prop:nls-nk}.


    {  For~\ref{prop:nls-lower_bound}, for $\ell\geq 2$, using a Taylor expansion around zero, for some $\xi\in (0,t)$ and $t<1-1/n$, we have
    \begin{align*}
    \Phi_{k,\ell}(t) &= \Gamma_k(-\ln Y_\ell (0)) + \Gamma_k(-\ln Y_\ell)'(0)t + \frac{1}{2}\Gamma_k(-\ln Y_\ell)''(\xi) t^2 \\
    & \geq (k-1)! - \frac{\theta_{\ell+1}^\star}{\theta_\ell^\star} k! t - \frac{c_k \ln(n)}{2}t^2 \tag{Using the previous part and $\Gamma_k(-\ln Y_\ell)(0)=0$ using NLS}\\
    &\geq (k-1)! - \frac{c_k \ln(n)}{2}t^2, 
    \end{align*}
    where we used the properties of $\nls_k(\theta^\star)$ and the definition of $b_k$. Since $Y_\ell(0)=1$, for some $\delta_k>0$ we have that $Y_\ell(t)=1-\varepsilon_\ell(t)$ for $t\in [0,\delta]$ with $\varepsilon_\ell(t)\leq 1/2$ for $t\in [0,\delta]$ and $\varepsilon_\ell(t)\to 0$ when $t\to 1$. We simply write $\varepsilon=\varepsilon_\ell(t)$ for convenience. Then, using the characterization of the gamma function $\Gamma_k$ as a Poisson distribution, we can deduce that
    \begin{align*}
    \frac{c_k \ln(n)}{2}t^2 &\geq \int_0^{-\ln Y_\ell(t)} s^{k-1} e^{-s}\, \mathrm{d}s \\
    &\geq \int_0^{\varepsilon} s^{k-1}e^{-s}\, \mathrm{d}s \\
    &= (k-1)! \sum_{j\geq k} e^{-\varepsilon} \frac{\varepsilon^j}{j!}\\
    &\geq \frac{(k-1)!}{k!} \varepsilon^k e^{-\varepsilon}\geq \frac{1}{2}\frac{(k-1)!}{k!} \varepsilon^k,    
    \end{align*}
    where in the second inequality we used that $\ln (1-\varepsilon)\leq -\varepsilon $ and the other inequalities follow by straightforward computations. From here, we obtain that $\varepsilon \leq \bar{c}_k \ln(n)^{1/k}  t^{2/k}$, where $\bar{c}_k = ( kc_k )^{1/k}$. This concludes~\ref{prop:nls-lower_bound}.   }
\end{proof}

\begin{proof}[Proof of Claim~\ref{claim:gamma_kk}] Using a Taylor expansion, we have
    \begin{align*}
        \Phi_{k,k}'\left( \frac{t}{\barnk} \right) - \frac{\Phi_{k,k}((t+1)/\barnk)) - \Phi_{k,k}(t/\barnk)}{1/\barnk} & = -\frac{1}{2\barnk} \Phi_{k,k}''(\xi) \tag{For some $\xi\in (t/\barnk, (t+1)/\barnk)$}
    \end{align*}
    We have $\Phi_{k,k}'' = - \Gamma_{k+1}(-\ln y_k)' = - (-\ln y_k)^{k} y_k'\geq 0$. This concludes the proof of the claim. 
\end{proof}

\begin{proof}[Proof of Claim~\ref{claim:phi-log}]
Using a Taylor expansion, we have
    \begin{align*}
    &\barnk(\Phi_{k,\ell}((t-1)/\barnk)-\Phi_{k,\ell}(t/\barnk)) + \Phi'_{k,\ell}((t-1)/\barnk)\\
    &=\barnk\left(\frac{1}{\barnk}\Phi_{k,\ell}'\left( \frac{t-1}{\barnk}\right)+\Phi_{k,\ell}\left(\frac{t-1}{\barnk}\right)-\Phi_{k,\ell}\left(\frac{t}{\barnk}\right)\right)=-\barnk\cdot \frac{1}{2\barnk^2} \Phi_{k,\ell}''(\xi)=-\frac{1}{2\barnk} \Phi_{k,\ell}''(\xi),
    \end{align*}
    for some value $\xi\in ((t-1)/\barnk, t/\barnk)$.
    Since $t/\barnk\le (\barnk-1)/\barnk=1-1/\barnk$, by Proposition \ref{prop:mon-phi}\ref{prop:nls-nk} we have $-\Phi_{k,\ell}''(\xi)\le c_k\ln(\barnk)$, which concludes the proof of the claim.
\end{proof}

\begin{proof}[Proof of Claim~\ref{claim:UB_for_terms_}]
    We first verify that for every $\ell \in [k]$, for $\barnk \geq 1/d_k$, and $t\leq \barnk-1$ we have $\varepsilon_{\ell,t} \leq 3\ln (\barnk)/\barnk$, where $d_k$ is defined in Proposition~\ref{prop:mon-phi}. Indeed, using Proposition~\ref{prop:mon-phi}\ref{prop:nls-delta} we obtain
    \begin{align*}
        -\ln Y_j(1-1/\barnk) & \leq -\ln(d_k)/\barnk + 2\ln(\barnk)/\barnk. 
    \end{align*}
    For $\barnk\geq 1/d_k =4/\min\{ \theta_{\ell+1}^\star/\theta_{\ell}^\star -1: { \ell \in \{1,\ldots,k-1\}}\}$, we obtain the desired result.
    Then
    \begin{align*}
        \left( 1- 4\frac{(k+1)^2}{n+1} \right)^{-1}\left( 1-\frac{\barnk \varepsilon_{j+1,t}^2}{1-\varepsilon_{j+1,t}}\right)^{-1} & \leq \left( 1 + 40 \frac{k^2}{n} \right) \left( 1 - \frac{\ln(\barnk)^2}{\barnk - \ln(\barnk)} \right)^{-1}\\
        &\leq 1 + 10 \frac{\ln (\barnk)^2}{\barnk},
    \end{align*}
    which holds for $n$ large.
\end{proof}

\section{Numerical Upper Bounds for Small $k$}\label{app:numerical_k_2}

In this section, for every $n\geq 1$, we provide a finite-dimensional LP with an optimal value that upper bounds the approximation ratio $\gamma_{n,k}$. Thus, solving this LP provides a venue to produce upper bounds on the worst-case approximation ratio, $\inf_{n\geq 1}\gamma_{n,k}$. Our construction is surprisingly simple as it only requires to provide a finite collection of points in $[0,1]$.
We use this LP to provide a new upper bound for small values of $k\in \{1,\ldots,5\}$; a summary is presented in Table~\ref{tab:table_2}. This complements our results in Table~\ref{tab:table_1}.
Let $0=q_0< q_1 < \cdots < q_m=1$ be a collection of points in $[0,1]$ and $Q=(q_0,\ldots,q_m)$. Consider the following finite-dimensional LP:
\begin{align}
	\min \quad\,\,  &d_{1,k} \tag*{\normalfont{\mbox{[P$(Q)$]$_{n,k}$}}}\label{form:LP_dual_Q}\\
	\normalfont\text{s.t.} \quad &d_{t,\ell} \geq \sum_{j=1}^m a_j \min\{q_j, q_i \} +   q_i d_{t+1,\ell-1} + (1-q_i)d_{t+1,\ell}, \, \text{ for } t\in[n], \ell\in[k], i \in \{0\} \cup [m], \label{const:dynamic_constr_dual_discrete} \\
	&\sum_{j=1}^m a_j \sum_{\ell=1}^k \Pr[\mathrm{Binom}(n,q_j)\geq\ell ] \geq 1,\label{const:max_value_const_dual_discrete} \\
	&a_j\geq 0,\qquad\qquad  \text{ for every } j\in \{0,\ldots,m\},\label{const:nonincreasing_dual_discrete}\\
    &d_{t,\ell}\ge 0, \qquad\qquad  \text{for every }t\in [n+1]\text{ and every }\ell\in \{0,1,\ldots,k\},\label{const:d-positive-dual_discrete}
\end{align}
Let $\gamma_{n,k}(Q)$ be the optimal value of~\ref{form:LP_dual_Q}. Our first result is a general methodology to bound $\gamma_{n,k}$ using~\ref{form:LP_dual_Q}. 

\begin{proposition}
    For any $Q=( q_0, q_1,\ldots, q_m)$ collection of points in $[0,1]$ with $q_0=0$, $q_m=1$, and for any $n\geq 1$, we have $\gamma_{n,k}(Q)\geq \gamma_{n,k} $. 
\end{proposition}

\begin{proof}
    It is enough to show that any solution $(a,d)$ to~\ref{form:LP_dual_Q} induces a feasible solution to~\ref{form:LP_dual} with objective value $d_{1,k}$. Let $h:[0,1]\to \R_+$ defined via 
    \[
    h(u) = \begin{cases}
        \sum_{j=0}^m a_j, & u=0,\\
        \sum_{j=i}^m a_j, & u\in (q_{i-1},q_i], i\in \{1,\ldots,m\}.
    \end{cases}
    \]
    Clearly, $h$ is non-negative and non-increasing. Note that for $q\in (q_{i-1},q_i]$, we have
    \begin{align*}
        \int_0^q h(u)\, \mathrm{d}u & = (q_1-q_0)\sum_{j=1}^m a_j+ \cdots + (q_{i-1}-q_{i-2})\sum_{j=i-1}^m a_j + (q-q_{i-1})\sum_{j=i}^m a_j\\
        & = \sum_{j=1}^{i-1} a_j q_j + \sum_{j=i}^m a_j q\\
        & = \sum_{j=1}^{m} a_j \min\{ q_j, q\}.
    \end{align*}
    Now, note that the function $q\mapsto \sum_{j=1}^m a_j \min \{q_j,q \} + (1-q)d_{t+1,\ell} + q d_{t+1,\ell-1}$
    is concave and piece-wise linear; hence, it attains its maximum in one of the breakpoints $q_0,\ldots,q_m$ when $q\in [0,1]$. This implies, that for any $q\in [0,1]$, and $\ell\in [k]$,
    \begin{align*}
        \int_0^q h(u)\, \mathrm{d}u + (1-q)d_{t+1,\ell} + q d_{t+1,\ell-1} & = \sum_{j=1}^m a_j \min \{q_j,q \} + (1-q)d_{t+1,\ell} + q d_{t+1,\ell-1}\\
        & \leq \max_{i\in \{0,\ldots,m\}} \left\{ \sum_{j=1}^m a_j \min \{q_j,q_i \} + (1-q_i)d_{t+1,\ell} + q_i d_{t+1,\ell-1} \right\} \\
        & \leq d_{t,\ell}
    \end{align*}
    where in the first equality we used the formula found for $\int_0^q h(u)\, \mathrm{d}u$, in the first inequality we used the fact in the previous paragraph and in the last inequality we used~\eqref{const:dynamic_constr_dual_discrete}. Furthermore,
    \begin{align*}
    \int_0^1 g_{n,k}(u) h(u)\, \mathrm{d}u & = \sum_{i=1}^m \int_{q_{i-1}}^{q_i} \sum_{j=i}^m a_j g_{n,k}(u)\, \mathrm{d}u\\
    & = \sum_{j=1}^m a_j \int_0^{q_j}g_{n,k}(u)\, \mathrm{d}u\\
    & = \sum_{j=1}^m a_j \sum_{\ell=1}^k \Pr[\mathrm{Binom}(n,q_j)\geq \ell]  \geq 1
    \end{align*}
    where in the equalities we use the definition of $f$, change the order of summation and performed the integral of $g_{n,k}$. The last inequality follows by~\eqref{const:max_value_const_dual_discrete}. From here, we obtain immediately that $(h,d)$ is feasible to~\ref{form:LP_dual} which concludes our proof.
\end{proof}

Numerically, we found that collections of points of the form $Q = ((i/100m)_{i \in \{0,1,\ldots,m\}},1)$ provide improved upper bounds compared to those by \cite{jiang2022tightness}. Table~\ref{tab:table_2} in Section~\ref{sec:final_remarks} summarizes our new numerical upper bounds for $n=1000$ and $m=1000$.

\end{document}